\documentclass[letterpaper,11pt]{article}

\usepackage[utf8]{inputenc}
\usepackage[T1]{fontenc}    
\usepackage{hyperref}       
\usepackage{url}            
\usepackage{mathtools}
\usepackage{amsmath}
\usepackage{amssymb}
\usepackage{amsthm}
\usepackage{thmtools}
\usepackage{thm-restate}
\usepackage{color}
\usepackage{hyperref}
\usepackage{array,multirow,multicol}
\usepackage{cleveref}
\usepackage[bold,full]{complexity}
\usepackage{todonotes}
\usepackage{wrapfig}
\usepackage{paralist}
\usepackage{xspace}
\usepackage{empheq}
\usepackage[figuresleft]{rotating}
\usepackage[thinlines]{easytable}
\usepackage{tikz}
\usetikzlibrary{shapes.geometric}
\usetikzlibrary{decorations.pathmorphing,decorations.pathreplacing,decorations.shapes}
\usepackage{wrapfig}
\usepackage[letterpaper,margin=1in]{geometry}
\usepackage{forest}
\usetikzlibrary{arrows.meta, positioning}
\usetikzlibrary{positioning, shapes.geometric}
\usepackage{bm}

\theoremstyle{definition}
\newtheorem{theorem}{Theorem}

\newtheorem{corollary}{Corollary}

\newtheorem{fact}[theorem]{Fact}
\newtheorem{proposition}[theorem]{Proposition}
\newtheorem{example}{Example}

\declaretheorem[
  name=Claim,
  numberwithin=proposition,
  refname={Claim,Claims},
  Refname={Claim,Claims}
]{claim}
 
\newcommand{\N}{\mathbb{N}}

\newcommand{\Q}{\mathbb{Q}}

\DeclareMathOperator{\im}{im}

\DeclareMathOperator{\rank}{rk}

\DeclareMathOperator{\GL}{GL}

\DeclareMathOperator{\diag}{diag}

\DeclareMathOperator{\id}{Id}
\newcommand{\Zcl}[2][]{\overline{#2}^{#1}}

\newcommand{\ExtAlg}[2][]{\Lambda^{#1}{#2}}

  {\par\addvspace{.6pc plus .2pc minus .1pc}
   \noindent
   {\bf\proofname}\enspace\ignorespaces}%
  {\par\addvspace{.6pc plus .2pc minus .1pc}}
\def\proofname{Proof.}
\def\qed{\relax\ifmmode\hfill \Box\else\unskip\nobreak\hfill $\Box$\fi}

\newenvironment{proofclaim}{%
  \par\noindent{Proof of Claim.}\enspace\ignorespaces
}{%
  \hfill$\clubsuit$\par\addvspace{.6pc plus .2pc minus .1pc}
}

\allowdisplaybreaks

\title{Algebraic Closure of Matrix Sets Recognized  by 1-VASS}

\date{}

\author{
 Rida Ait El Manssour\\
   University of Oxford, UK
  \and
  Mahsa Naraghi\\
  Universit\'e Paris Cit\'e, IRIF, France
  \and
  Mahsa Shirmohammadi\\
  CNRS, IRIF, France
  \and
  James Worrell\\
  University of Oxford, UK
}

\begin{document}
\maketitle
\begin{abstract}
It is known how to compute the Zariski closure of a finitely generated monoid of matrices and, more generally, of a set of matrices specified by a regular language.
This result was recently used to give a procedure to compute all polynomial invariants of a given affine program.
Decidability of the more general problem of computing all polynomial invariants of affine programs with recursive procedure calls remains open.
Mathematically speaking, the core challenge is 
to compute the Zariski closure of a set of matrices defined by a context-free language.   In this paper, we approach the problem from two sides:
Towards decidability, we give a procedure to compute the Zariski closure of sets of matrices given by one-counter languages (that is, languages accepted by one-dimensional vector addition systems with states and zero tests), 
a proper subclass of context-free languages.
On the other side, we show that the problem becomes undecidable for indexed languages, a natural extension of context-free languages corresponding to nested pushdown automata.
One of our main technical tools is a novel adaptation of Simon's factorization forests to infinite monoids of matrices.
\end{abstract}




\section{Introduction}

\subsection{Context}
Finitely generated matrix monoids are not recursive in general, and many of the associated computational problems are undecidable.
In particular, the Mortality Problem (determine whether a given finitely generated matrix monoid contains the zero matrix) is undecidable already in dimension~3~\cite{paterson1970unsolvability}.
On the other hand, it is known how to decide certain global properties of matrix semigroups. For example, one can determine
whether a finitely generated sub-monoid of the monoid $M_d(\mathbb Q)$  of $d\times d$ rational matrices is finite~\cite{BumpusHKST20,Jacob77,MandelS77}. More generally,
there is a procedure to compute the Zariski closure of finitely generated matrix monoids~\cite{HrushovskiOPW23}.
The output  of this algorithm is a finite collection of polynomials that capture all the algebraic relationships between the different entries of each matrix in the monoid
(such as the property that the determinant of every matrix equals one).  The set of common zeros of the latter collection of polynomials is the Zariski closure of the monoid.

The algorithm for computing the Zariski closure of a matrix monoid admits a straightforward generalization that, given a regular language $L\subseteq \Sigma^*$
and a monoid morphism $\varphi:\Sigma^*\rightarrow M_d(\mathbb Q)$, allows computing the Zariski closure of the set $\varphi(L)$ in $M_d(\mathbb Q)$.  
Such a procedure was used in~\cite{HrushovskiOPW23} to resolve a problem posed by  M\"{u}ller-Olm and Seidl~\cite{Muller-OlmS04} 
on the computability of polynomial invariants in program analysis.  A special case of this result, concerning the linear Zariski closure, was recently used by Bell and Smertnig~\cite{BellS23} to resolve 
a longstanding open problem concerning ambiguity in formal languages and automata. 
The procedure of~\cite{HrushovskiOPW23} generalizes and builds on a procedure of~\cite{DerksenJK05} which solves the version of the problem for matrix groups and which was applied to study language emptiness for quantum automata.  There is a wealth of related work on the automatic synthesis of program invariants and procedure summaries; a select sample includes~\cite{CyphertK24,HumenbergerJK18,KincaidCBR18,SankaranarayananSM04,ManssourKSV25}. 

In this paper we are interested in computing the Zariski closure of sets of matrices $\varphi(L)$ for languages $L$ beyond regular.  
One of our main motivations is an application to interprocedural program analysis, where the problem of computing all polynomial invariants 
for a collection of affine programs that feature recursive procedure calls (as studied in~\cite{Chatterjee0GG20,gulwani2007computing,GodoyT09,Muller-Olm04c,Muller-OlmPS06}, among many other sources) can be
understood in terms of computing the Zariski closure of $\varphi(L)$ for $L$ a context-free language.
In the presence of recursive procedure calls, regular languages no longer suffice to model all paths through a control-flow graph that models a program.
Figure~\ref{fig:RECUR} shows two recursive affine programs, in the sense of~\cite{Muller-Olm04c}, in which the the set of execution paths is defined by a one-counter language.
Affine programs feature integer-valued variables that can be updated by linear assignments.  In the interests of decidability, the boolean
conditionals in loop guards and control-flow statements are abstracted to wildcards $(*)$  that non-deterministically evaluate to true or false.  At present it is  open whether one can compute all interprocedural polynomial invariants, not just those of an \emph{a priori} bounded degree, in the setting of affine programs (see the discussion in ~\cite[Section 8]{Muller-OlmPS06}).

\begin{figure}[h]
\centering

\begin{minipage}{0.45\textwidth}
\begin{tabbing}
\textsf{procedure} $P()$ \\
\hspace{1em}\=\textsf{begin} \\
\>\hspace{1em}\=\textsf{if} $(\ast)$ \textsf{then} \\
\>\>\hspace{1em}\=$\boldsymbol{x} := A \boldsymbol{x}$ \\
\>\>\>\textsf{call} $P()$ \\
\>\>\>$\boldsymbol{x} := B \boldsymbol{x}$ \\
\>\>\textsf{else} \\
\>\>\>\textsf{skip} \\
\>\>\textsf{endif} \\
\>\textsf{end}
\end{tabbing}
\end{minipage}
\qquad\qquad\qquad
\begin{minipage}{0.45\textwidth}
\begin{tabbing}
\textsf{procedure} $Q()$ \\
\hspace{1em}\=\textsf{begin} \\
\>\hspace{1em}\=\textsf{if} $(\ast)$ \textsf{then} \\
\>\>\hspace{1em}\=$\boldsymbol{x} := A \boldsymbol{x}$ \\
\>\>\>\textsf{call} $Q()$ \\
\>\>\>$\boldsymbol{x} := B \boldsymbol{x}$ \\
\>\>\>\textsf{call} $Q()$\\
\>\>\textsf{else} \\
\>\>\>\textsf{skip} \\
\>\>\textsf{endif} \\
\>\textsf{end}
\end{tabbing}
\end{minipage}

\caption{Two recursive procedures in which all variables are global.  The wildcard (*) evaluates non-deterministically either to true or false.  An execution of the procedure $P$ on the left can be summarized by the assignment $\boldsymbol x:= B^nA^n \boldsymbol x$ for some $n\in \mathbb N$.  An execution of the procedure $Q$ on the right results in assignment~$\boldsymbol x:=M\boldsymbol x$ where the matrix $M$ is a product of a string of $A$'s and $B$'s that is well-matched in the sense of the Dyck language.}
\label{fig:RECUR}
\end{figure}

\subsection{Main Results}
The main decidability results in the paper concern
one-counter languages, which are a special case of context-free languages.
In fact, our technical development uses the terminology and formalism of 
1-VASS (one-dimensional vector addition systems with states).
A 1-VASS is a non-deterministic finite automaton that is equipped with a single counter that takes values in the nonnegative integers and that can be incremented and decremented on transitions.  Decrement transitions cannot be taken when the counter is zero.
Acceptance is either by control state, in which case we speak of the \emph{coverability language}, or by control state and counter value, in which case we speak of the \emph{reachability language}.  Perhaps the best known example of a 1-VASS reachability language is the Dyck language of well-parenthesized words on a binary alphabet.
A one-counter automaton is a 1-VASS with additional transitions, called \emph{zero tests}, that can only be taken  when the counter is zero.  As we will see, it is straightforward to extend our decidability results to one-counter automata.  

Let $\varphi : \Sigma^*\rightarrow M_d(\mathbb Q)$ be a monoid morphism.
Given a 1-VASS coverability or reachability language~$L\subseteq \Sigma^*$,
we are interested in computing the set of all polynomials in $\mathbb Q[\{x_{ij}\}_{1\leq i,j\leq d}]$ 
that vanish on the set $\varphi(L)=\{\varphi(w) : w\in L\}$  of matrices.  
In general, this set of polynomials forms an ideal in the polynomial ring 
$\mathbb Q[\{x_{ij}\}_{1\leq i,j\leq d}]$ and the goal is to compute a finite basis of this ideal
(Such a finite basis necessarily exists by the Hilbert basis theorem).
The set of common zeros of these polynomials is the \emph{Zariski closure} of 
$\varphi(L)$.  We regard this as a subset of $M_d(\overline{\mathbb Q})$,
where $\overline{\mathbb{Q}}$ is the set of algebraic numbers (that is, complex numbers that are roots of univariate polynomials in~$\mathbb{Q}[x]$).
The Zariski closure is thus  an  over-approximation of the original set, capturing all polynomial relations satisfied by the entries of its matrices.

The first main contribution of the paper is a procedure  that takes as 
input  a 1-VASS coverability language $L$ on alphabet $\Sigma$ and a morphism
    $\varphi:\Sigma^*\rightarrow M_d(\mathbb Q)$, and outputs a finite collection of polynomials 
    that generates the vanishing ideal of $\varphi(L)$.   
We also present a more general procedure for computing the vanishing ideal of $\varphi(L)$ for~$L$ a 1-VASS reachability language.
We complement the above two results by showing  undecidability of computing the Zariski closure of $\varphi(L)$, where $L$ is an indexed language.
Recall that indexed grammars~\cite{Aho68} extend context-free grammars by allowing stack-like memory on nonterminals.  Indexed grammars correspond to second-order pushdown automata,
in which the main stack can hold secondary stacks as element.  
A classic example of an indexed language that is not context-free is $\{a^nb^nc^n : n\in \mathbb N\}$.

\subsection{Examples}
The following examples illustrate the above results.
The steps of the computations  are shown in~\Cref{app:computation}.
Recall that a 1-VASS is a non-deterministic finite automaton in which each transition is annotated by an integer weight.
The \emph{coverability language} is the set of all words arising from an accepting run of the automaton in which 
every prefix has non-negative weight.
The \emph{reachability language} is the set of all words arising from accepting runs of total weight zero in which every prefix has non-negative weight.
In the examples below, the state $s$ is initial and the state $t$ is accepting.


\begin{example}
\label{ex:anbndyck}
Consider the following 1-VASS with  the reachability language~$L_R=\{a^nb^n \mid n\in \mathbb{N}\}$, and the monoid morphism~$\varphi : \{a,b\}^*\rightarrow M_2(\mathbb Q)$ as defined below. 
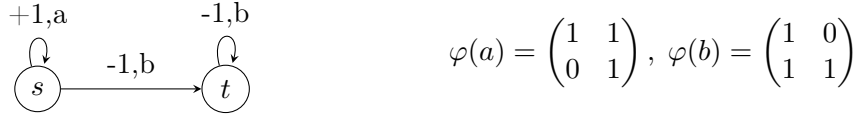
\begin{center}
\begin{minipage}{0.5\textwidth}
\centering
       \begin{tikzpicture}[->,>=stealth,auto,node distance=2.5cm]
  \node (s) [circle,draw] {$s$};
  \node (t) [circle,draw,right of=s] {$t$};

  \path (s) edge[loop above] node[above] {+1,a} (s)
            edge  node[above] {-1,b} (t)
        (t) edge[loop above] node[above] {-1,b} (t);
\end{tikzpicture}
\end{minipage}
\begin{minipage}{0.4\textwidth}
  $\varphi(a)=\begin{pmatrix}
    1 & 1\\
    0 & 1
\end{pmatrix},\; 
\varphi(b)=\begin{pmatrix}
    1 & 0\\
    1 & 1
\end{pmatrix}$
\end{minipage}
\end{center}
The Zariski closure 
$\overline{\varphi(L_R)} \subseteq M_2(\overline{\mathbb Q})$ is the zero set of 
 the ideal $I_R \subseteq \mathbb Q[x_{11},x_{12},x_{21},x_{22}]$ generated by the  polynomial relations between entries of the matrices in $\varphi(L_R)$.  
As the image of the  word~$a^nb^n$ under $\varphi$ is~$\begin{pmatrix}
    n^2+1 & n\\
    n & 1
\end{pmatrix}$, by simple algebraic manipulations one obtains that
\[\langle x_{11}-x_{12}x_{21}-1, \, x_{12}-x_{21},\, x_{22}-1\rangle.
\]
Similarly, the Zariski closure 
$\overline{\varphi(L_C)} \subseteq M_2(\overline{\mathbb Q})$ is the zero set of 
 $\langle x_{11}-x_{12}x_{21}-1, \, x_{22}-1\rangle$.
\hfill $\blacktriangleleft$
\end{example}

\begin{example}\label{ex:Dyck}
A second simple example concerns the so-called Dyck language~$L$, consisting of all words with same number of $a$'s and $b$'s such that every prefix 
has at least as many $a$'s and $b$'s.
This is the reachability language of the following 1-VASS:
\begin{center}
       \begin{tikzpicture}[->,>=stealth,auto,node distance=2.5cm]
  \node (s) [circle,draw] {s};
  
  \path (s) edge[loop left] node[left] {+1,a} (s)
            
        (s) edge[loop right] node[right] {-1,b} (s);
\end{tikzpicture}
\end{center}
Since $L$ is closed under concatenation, its image under the above morphism $\varphi$ is closed under multiplication.  
Moreover, since $\varphi(a)$ and $\varphi(b)$ are invertible matrices, $\varphi(L)$ is a sub-semigroup of $GL_2(\mathbb Q)$. Due to \Cref{fact:closed_semigroup_is_subgroup}, the Zariski closure of $\varphi(L)$  inside  $GL_2(\overline{\mathbb Q})$ is a linear algebraic subgroup. 
Since both $\varphi(a)$ and $\varphi(b)$ have determinant one,
all elements 
$\varphi(L)$ also have determinant one.  In this case, the ideal of all polynomials relations satisfied 
by~$\varphi(L)$
is the principal ideal generated by $x_{11}x_{22}-x_{12}x_{21}-1$, expressing that the determinant is one.
In other words, the
Zariski closure is the special linear group 
$SL_2(\overline{\mathbb Q})$ of all $2\times 2$ matrices with algebraic entries and determinant one. \hfill $\blacktriangleleft$
\end{example}

\begin{example}\label{ex:RvsC2}
Consider the following 1-VASS, whose coverability language~$L_C$ is a subset of the regular 
language $a^*(a^*b^2)^* b$. For every word $w \in L_C$, every prefix of $w$ contains at least as many occurrences of $a$ as of~$b$.
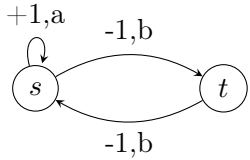
\begin{center}
\begin{minipage}{0.5\textwidth}
\centering
   \begin{tikzpicture}[->,>=stealth,auto,node distance=2.5cm]
  \node (s) [circle,draw] {$s$};
  \node (t) [circle,draw,right of=s] {$t$};

  \path (s) edge[loop above] node[above] {+1,a} (s)
            edge[bend left]  node[above] {-1,b} (t)
        (t) edge[bend left] node[below] {-1,b} (s);
\end{tikzpicture}
\end{minipage}
\end{center}

 The reachability language~$L_R$ of this 1-VASS is the subset of $L_C$
 determined by the condition that the total number of $a$'s and $b$'s in a word be equal.
For the morphism~$\varphi_1$ and ideal $I$ shown below, the respective closures of $\varphi_1(L_C)$ and $\varphi_1(L_R)$
coincide and are the zero set of~$I$:
\begin{center}
\begin{minipage}{0.5\textwidth}
  $\varphi_1(a)=\begin{pmatrix}
    2 & 0\\
    0 & 4
\end{pmatrix},\; 
\varphi_1(b)=\begin{pmatrix}
    1 & 0\\
    1 & 1
\end{pmatrix}$
\end{minipage}
\begin{minipage}{0.4\textwidth}
\centering
  $ I:=\langle x_{12}, \, x_{11}^2-x_{22}\rangle$.
\end{minipage}
\end{center}
As a final example
consider the set of matrices $\varphi_2(L_C)$
arising from the morphism~$\varphi_2$, shown below:
\[
  \varphi_2(a)=\begin{pmatrix}
    1 & 1 & 0\\
    0 & 1& 1\\
    0 & 0 & 1
\end{pmatrix},\; 
\varphi_2(b)=\begin{pmatrix}
    1 & -1 & 0\\
    0 & 1& 1\\
    0 & 0 & 1
\end{pmatrix}
\]
The closure of $\varphi_2(L_C)$ is the zero set of~$I_C:=\langle x_{11} - 1, x_{22} - 1, x_{33} - 1, x_{21},  x_{31},x_{32}\rangle$.
The defining ideal of the closure of $\varphi_2(L_R)$ is
$\langle I_C, x_{12}\rangle$, where
the additional equation $x_{12}=0$
reflects  the condition that 
 total number of~$a$'s and $b$'s in every word $w\in L_R$ is equal.
\hfill $\blacktriangleleft$
\end{example}

\subsection{Discussion}
The existence of a procedure for computing the Zariski closure of $\varphi(L)$, 
given a context-free language $L$ and a morphism $\varphi$, remains open.
Extending beyond interprocedural program analysis to models of concurrency and parameterized systems, 
another natural next step is to consider the case where $L$ is the coverability or reachability language of a VASS in dimension greater than one.
We note that Presburger invariants for VASS reachability sets were studied in~\cite{jerome2010general}, leading to a new algorithm for solving VASS reachability~\cite{leroux2011vector}. (The VASS reachability problem has attracted  renewed interest following the resolution of a long-standing question concerning its complexity~\cite{CzerwinskiLLLM19,LerouxS19,CzerwinskiO21,Leroux21}.)
Comparing the result of~\cite{jerome2010general} to the  generalization of \textsc{Reach Closure} to general VASS, we note on the one hand that 
the reachable set of configurations of a given VASS can easily be expressed as a morphic image of its reachability language but, on the other hand, 
Presburger invariants involve linear inequalities whereas \textsc{Reach Closure} concerns polynomial equations.

Our work also has similarities to a line of research  on
regular separability of VASS languages~\cite{CzerwinskiLLLM19,CzerwinskiZ20}.  
Here, the question is to determine, given two VASS, whether there is a regular language that contains the language of one VASS and is disjoint from the language of the other. For comparison, our solution to \textsc{Cover Closure}
involves computing a regular language that contains a given 1-VASS language and has the same Zariski closure.


\section{Summary and Overview}
\label{sec-overview}
This section contains statements of the main results in the paper and a
high-level overview of the technical approach.  The detailed
development is given in the remainder of the paper.
\subsection{Stateless Setting}
To study the problems of computing the Zariski closure of (the morphic images) of 1-VASS coverability 
and reachability languages in an abstract setting, we introduce two families of languages.
Given a morphism $\omega:\Sigma^*\rightarrow \mathbb Z$
we define the language $L_C(\omega)$, comprising those words whose
prefixes all have nonnegative $\omega$-value, and the language
$L_R(\omega)$, comprising those words in $L_C(\omega)$ that also have
value zero.  Formally we have:
\begin{align*}
L_C(\omega)&\,:=\,\{w\in\Sigma^*: \forall u\preceq w\, (\omega(u)\geq 0)\}\\
L_R(\omega)&\,:=\, \{ w\in \Sigma^*: \omega(w)=0 \wedge \forall
                      u\preceq w\, (\omega(u)\geq 0)\} \, .
\end{align*}
where $\preceq$ denotes the prefix relation on words.
The core of our technical development addresses the following two decision problems:
given morphisms $\varphi:\Sigma^*\rightarrow M_d(\Q)$ and $\omega:\Sigma^*\rightarrow \mathbb Z$, the
\textsc{Cover Closure} problem asks to compute the Zariski
closure of $\varphi(L_C(\omega))$, while
the \textsc{Reach Closure} problem asks to compute the Zariski closure of $\varphi(L_R(\omega))$.
We then have:
\begin{restatable}{proposition}{VASStoMON}\label{pro:VASStoMON}
    The problem of computing the Zariski closure of a 1-VASS coverability language is interreducible with 
    \textsc{Cover Closure}.  The problem of computing the Zariski closure of a 1-VASS reachability language is interreducible with
    \textsc{Reach Closure}.
\end{restatable}
In the  proof of Proposition~\ref{pro:VASStoMON}, given in Section~\ref{sec:VASStoMON}, 
the arguments for coverability and reachability are essentially the same; here we sketch the former.
In brief, the problem of computing the closure of a 1-VASS coverability language can be reduced to the special case of 1-VASS
with a single control state by blowing up the dimension of the matrices.  Second, the version of the problem for single-state 1-VASS
is equivalent to \textsc{Cover Closure} since every language $L_C(\omega)$ is the coverability language of a single-state 1-VASS and, conversely, every
such coverability language is the homomorphic image of a language of the form $L_C(\omega)$.
As noted in Section~\ref{sec:VASStoMON}, a corollary of Proposition~\ref{pro:VASStoMON} is that in showing decidability 
of \textsc{Cover Closure} and \textsc{Reach Closure} we may assume without loss of generality that 
every symbol $\sigma\in\Sigma$ has weight $\omega(\sigma)\in\{-1,0,+1\}$.
We make this assumption henceforth.





A 1-VASS $\mathcal V$ with \emph{zero tests} features transitions that can only be taken when the counter is zero.
The reachability language of $\mathcal V$ can be effectively represented by a regular expression whose atoms 
are reachability languages of 1-VASS without zero tests (decompose each accepting run of $\mathcal V$ into segments between consecutive 
zero-tests).  This observation allows us to reduce
 the problem of computing $\Zcl{\varphi(L(\mathcal V))}$ 
 corresponding problem for 1-VASS without zero tests 
 which, by Proposition~\ref{pro:VASStoMON}, is equivalent to \textsc{Reach Closure}.

\subsection{Invertible Matrices} 
In the special case that the morphism $\varphi: \Sigma^* \rightarrow M_d(\mathbb{Q})$ takes values in the group of invertible matrices,
we can use a Noetherian property for ascending chains of linearly algebraic groups, due to~\cite{nosan2022computation}, to show that 
\textsc{Cover Closure} and \textsc{Reach Closure} are both computable.
To explain this, we first note that by~\Cref{fact:closed_semigroup_is_subgroup}, the condition on $\varphi$ entails that 
$\Zcl{\varphi(L_C(\omega))}$ and $\Zcl{\varphi(L_R(\omega))}$ are both linear algebraic groups (not just monoids).
Now~\cite[Theorem 11]{nosan2022computation} states that every strictly increasing chain 
$G_1\subsetneq G_2\subsetneq \cdots\subsetneq G_\ell \subseteq \ldots $ of Zariski-closed groups of invertible matrices in $M_d(\Zcl{\Q)}$ is finite, 
provided that each $G_i$ are topologically generated
by rational matrices.
(Here 
we say that $G$ is \emph{topologically generated} by a set $S$ if $G = \overline{\langle S \rangle}$.)

 We can use this Noetherian property to compute $\Zcl{\varphi(L_C(\omega))}$ and $\Zcl{\varphi(L_R(\omega)}$ as follows:
 First, we partition the alphabet $\Sigma$ into subsets $\Sigma_0$, $\Sigma_1$, and $\Sigma_{-1}$, consisting of letters with weights $0$, $1$, and $-1$, respectively. 
 There are then two cases:
\begin{enumerate}
\item \textsc{Cover Closure}: We define $G_0$ to be the group topologically generated by $\varphi(\Sigma_0 \cup \Sigma_1)$, and for all $i\in \mathbb{N}$, we write
\[G_{i+1} = \overline{\langle \varphi(\Sigma_1)\,  G_i \, \varphi(\Sigma_{-1})\, ,\,  G_i \rangle}\,.\]
\item \label{Inv Reach Closure} \textsc{Reach Closure}: We define $G_0$ to be the group
topologically generated by $\varphi(\Sigma_0)$ and for all~$i\in \mathbb{N}$,  similar to the previous case, take $G_{i+1}$ to be topologically generated by 
\begin{equation}\label{eq:ReachInv}
    G_{i+1} = \overline{\langle \varphi(\Sigma_1)\,  G_i \, \varphi(\Sigma_{-1})\, ,\,  G_i \rangle}\,.
\end{equation}
 \end{enumerate}

We will show that once these chains stabilize, the resulting group is
precisely the Zariski closure of the monoid of interest. Here we give the proof in the case of \textsc{Reach Closure}, the proof for \textsc{Cover Closure} follows similarly.
\begin{proposition}
    Consider the chain of groups $G_i$ constructed in \Cref{eq:ReachInv}, and let $m$ be the smallest integer such that $G_m = G_{m + 1}$. Then $\Zcl{\varphi(L_R(\omega))} = G_m.$
\end{proposition}
\begin{proof}
First, let us prove that if there are two consecutive equivalent  elements in the chain then the chain would stabilize. 
    Let $m$ be such that $G_m = G_{m + 1}$. Then we have that 
    \[
    G_{ m + 2} = \Zcl{ \langle \varphi(\Sigma_1)\,  G_{m + 1} \, \varphi(\Sigma_{-1}) , G_{m + 1} \rangle } =  \Zcl{ \langle \varphi(\Sigma_1)\,  G_{m } \, \varphi(\Sigma_{-1}) , G_{m } \rangle } = G_m\,.
    \]
    Hence, $G_{ m + k} = G_m$ for every~$k \in \N.$

    Next, we define a sequence of languages in recursive way as follows: 
    \begin{align*}
    L_0 = \Sigma_0 \qquad \text{and} \qquad
    L_{i + 1} = \left( \Sigma_1 L_i \Sigma_{-1} \cup L_i\right)^*\,.
     \end{align*}
By an induction on $i$, it follows that
 $G_i = \Zcl{\varphi(L_i)} $. Moreover, we have $L_R(\omega) = \bigcup_{ i \in \N} L_i$.  Thus, 
 \[
 \Zcl{\varphi(L_R(\omega))} = \Zcl{\bigcup_{i \in \N}  \Zcl{\varphi(L_i)}}
 = G_m.\]
\end{proof}

The given construction shows that the problems \textsc{Cover Closure} and \textsc{Reach Closure} are decidable in the case of invertible matrices. However, this  construction breaks down for non-invertible matrices, as shown by the following example.

\begin{example}
    Let $\Sigma=\{\alpha, \beta,\sigma_1,\sigma_2\}$. Define $\varphi: \Sigma^* \rightarrow M_2(\mathbb{Q})$ and  by 
    \[\varphi(\sigma_1)=
    \begin{pmatrix}
        0&0\\0&0
    \end{pmatrix}
    \qquad \varphi(\sigma_2)=\begin{pmatrix}
0 & 1\\
0 & 0
\end{pmatrix} \qquad \varphi(\alpha)=\begin{pmatrix}
 2 & 0\\
 0 & 1
 \end{pmatrix}  \qquad \varphi(\beta)=\begin{pmatrix}
 \frac{1}{2} & 0\\
 0 & 1
 \end{pmatrix} \, , \]
and $\omega: \Sigma^* \rightarrow \mathbb{Z}$ by \[\omega(\sigma_1)=\omega(\sigma_2)=0 \qquad \qquad \omega(\alpha)=-\omega(\beta)=1 \,.\]
Let $S_0$ be the  monoid topologically generated by $\varphi(\sigma_1)$ and $\varphi(\sigma_2)$. We write   $S_0:=\overline{\langle \varphi(\sigma_1),\varphi(\sigma_2)\rangle}$, and observe that~$S_0$ is a finite set containing the two generators and identity matrix. Consider the  chain 
$(S_i)_{i\geq 0}$
defined inductively by $S_{i+1} = \langle \varphi(\alpha)\,  S_i \, \varphi(\beta)\, ,\,  S_i \rangle$, as described above for \textsc{Reach Closure}. 
\emph{But this chain does not stabilize.}
 Indeed, we claim that for each $i \in \mathbb{N}$,
\[S_i=\Biggl\{ \begin{pmatrix}
        0&0\\0&0
    \end{pmatrix},\,\begin{pmatrix}
        1&0\\0&1
    \end{pmatrix},\,
\begin{pmatrix}
0 & 2^j\\
0 & 0
\end{pmatrix}; 0\leq j\leq i \Biggl\}\,.\]
 \sloppy The base case ($i = 0$) holds by definition of~$S_0$. For the inductive step, observe that
  $\begin{pmatrix}
2 & 0\\
0 & 1
\end{pmatrix}\begin{pmatrix}
0 & 2^j\\
0 & 0
\end{pmatrix}\begin{pmatrix}
1/2 & 0\\
0 & 1
\end{pmatrix}=\begin{pmatrix}
0 & 2^{j+1}\\
0 & 0
\end{pmatrix}$. Thus, each iteration introduces a new matrix with an entry $2^{j+1}$ in the upper-right corner, and so the chain $(S_i)_{i\geq 0}$ grows strictly at each step, never stabilizing. \hfill $\blacktriangleleft$
\end{example}

Solving the \textsc{Cover Closure} and  \textsc{Reach 
 Closure} in the general case---without the invertibility assumption---appears to be more challenging.
 The following section introduces the key technical tool that underlies our approach in the general setting.
\subsection{Factorization Trees for Matrix Monoids}
We say that a matrix~$M \in M_d(\mathbb Q)$ is \emph{stable} if $M$ and all its powers have the same rank.  An equivalent condition is that~$ \im(M) \cap \ker(M) = \{0\}$
where $\im(M)$ and $\ker(M)$ respectively denote the image and kernel of the linear transformation~$x \mapsto M x$.  
Using the notion of stable matrix we adapt the notion of factorization tree from finite monoids to the infinite monoid $M_d(\mathbb Q)$.  The original notion of factorization tree was introduced by Imre Simon~\cite{Simon90} in his resolution of the star-height problem for regular languages (see also~\cite{Bojanczyk09}).  In our variant, the role of idempotent elements in the classical version is replaced by stable matrices.  

We define a  factorization tree in $M_d(\mathbb Q)$  to be a finite  tree in which each node is labeled by an element of~$M_d(\mathbb Q)$, subject to two conditions:
\begin{itemize}
    \item  Product condition: If a node has children with labels $M_1, \ldots,M_k$ (read from left to right), then the label of the node is the matrix product $M_1 \cdots  M_k$.
    \item  Stability condition: If a node has more than two children, then each of its children must be a stable matrix. 
\end{itemize}

The \emph{height} of a factorization tree is defined as the maximum distance from the root to any leaf. 
The \emph{yield} of the tree is the sequence of leaf labels read from left to right, which is a sequence of matrices in $M_d(\mathbb Q)$.
Using manipulations in multilinear algebra we prove that every sequence of matrices in $M_d(\mathbb Q)$, no matter its length, can be realized as the yield of factorization tree of 
height bounded by the explicit function $\theta(d)$ of the dimension~$d$.

\begin{restatable}{proposition}{proptree}\label{prop:tree}
Every sequence of matrices in $M_d(\mathbb Q)$ is the yield of some  factorization tree of height at most~$\theta(d):=d(d+3)$.\end{restatable}

We use Proposition~\ref{prop:tree} to prove the following result,
which shows that all words of a sufficiently large positive weight contain a
stable factor with positive weight and, dually, all words of a
sufficiently large negative weight contain a stable factor with
negative weight.
\begin{restatable}{proposition}{decompose}\label{prop:eta}
There exists a threshold~$\eta(d):=2^{d(d+3)}+1$ such that for all $w \in \Sigma^*$:
\begin{itemize}
    \item[(i)] If $\omega(w) = \eta(d)$, then there exists a factor $u$ of $w$ such that
    $\omega(u)>0$ and $\varphi(u)$ is stable.
    \item[(ii)] If $\omega(w) = -\eta(d)$, then there exists a factor $u$ of $w$ such that
    $\omega(u)<0$ and $\varphi(u)$ is stable.
\end{itemize}
\end{restatable}

\subsection{Approximation by Regular Languages}
We use our results on factorization trees to reduce the \textsc{Cover
  Closure} and \textsc{Reach Closure} problems to the problem of
computing the Zariski closure of a set of matrices $\varphi(L)$ for a
regular language $L$ and morphism
$\varphi:\Sigma^*\rightarrow M_d(\mathbb Q)$, for which there is a
known algorithm~\cite{HrushovskiOPW23}.  This yields:

\begin{restatable}{theorem}{theNmon}
  \label{theo:coverVASS}
  The \textsc{Cover Closure} and \textsc{Reach Closure} problems
  are decidable.
 \end{restatable}

 Even though there is an easy reduction of \textsc{Cover Closure}
 to \textsc{Reach Closure}, it is instructive to give  separate
 proofs of the decidability of \textsc{Cover Closure}.
 
 At a high level, the proof strategy for showing decidability of
 \textsc{Cover Closure} is to compute an NFA $\mathcal A$ such that
 $L_{C}(\omega) \subseteq L(\mathcal A)$ and
 $\overline{\varphi(L_{C}(\omega))} = \overline{\varphi(L(\mathcal
   A))}$.  Automaton $\mathcal A$ keeps a count of the weight of the
input word up to a certain threshold.  If the weight of
 the input word ever exceeds the threshold the automaton moves to a
 sink state from which it accepts all words.  This means
 that up to the threshold, the automaton can check that the weight of
 the word read so far is not negative, but thereafter there are no
 guarantees on the weight of the prefixes.

By construction, $L(\mathcal A)$ is a superset of $L_{C}(\omega)$.
In the other direction, using Proposition~\ref{prop:eta}, we are able to show that 
every word $w \in L(\mathcal A) \setminus L_{C}(\omega)$ admits a
factorization $w=xuy$ such that $xu \in L_{C}(\omega)$, $\varphi(u)$
is stable, and $\omega(u)>0$.
It follows that for all sufficiently large $n$ we have $xu^{n+1} y \in
L_{C}(\omega)$ since by pumping the positive-weight segment we
assure that all prefixes have non-negative weight.
By the fact that $\varphi(u)$ is stable, the Zariski-closed monoid $\overline{\{ \varphi(u^n):n\in \mathbb N\}}$ is in fact a group.  In particular, this set
contains a matrix $I_u$ such that~$\varphi(u)\, I_u = I_u\, \varphi(u) = \varphi(u)$.  We conclude that 
$w \in \Zcl{\{ xu^{n+1} y : n \in \mathbb N\}}$ and hence $w \in \Zcl{L_{C}(\omega)}$.

The following simple but instructive example shows that for 
reachability languages there may not be regular superset that has the same
Zariski closure.
\begin{example}
Consider the alphabet $\Sigma = \{a,b\}$ and morphisms
$\varphi:\Sigma^*\rightarrow M_1(\mathbb Q)$
and $\omega:\Sigma^*\rightarrow \mathbb Z$ given by~$\varphi(a)=2,\varphi(b)=1/2$
and $\omega(a)=1,\omega(b)=-1$.
The reachability language $L_{R}(\omega)$ comprises all words in $\{a,b\}^*$ with the same number of $a$'s as
$b$'s.  Thus $\varphi(L_{R}(\omega))$
is the singleton $\{1\}$.  

By the Myhill-Nerode Theorem, for any regular language $L'$ such that 
        $L_{R}(\omega) \subseteq L'$,  
        there exists an infinite sequence $m_1 < m_2 < m_3 < \cdots$ yielding  the same left derivatives  $(a^{m_i})^{-1}L'$. By this and the fact that  $a^{m_i} b^{m_i} \in L'$, all words
$a^{m_i}b^{m_j}\in L'$ for all $i,j\geq 1$.  It follows that $\Zcl{\varphi(L')} = \overline{\mathbb Q}$.  Hence there is no regular 
    language that includes~$L_N(\omega)$ and has the same Zariski
    closure.

Note, however, that for any morphism $\varphi:\Sigma^*\rightarrow
M_d(\mathbb Q)$ there is a simple method to compute
$\Zcl{\varphi(L_R(\omega))}$. This set is the closure of the
image of  $\Zcl{\langle M \rangle}$ under the map
$\Psi:M_{2d}(\Zcl{\Q}) \rightarrow M_d(\Zcl{\Q})$, where 
\[ M = \begin{pmatrix}
\varphi(a) & 0\\ 0 & \varphi(b) 
  \end{pmatrix} \qquad\text{ and }\qquad 
\Psi\begin{pmatrix} X&0\\0&Y\end{pmatrix} = XY \, .
\]
 
  It thus suffices to compute the Zariski closure of the cyclic 
  monoid $\langle M \rangle$, which is well understood (see, e.g.,~\cite{DerksenJK05}).
\label{ex:simple}
\hfill $\blacktriangleleft$
  \end{example}

  Our method to show decidability of \textsc{Reach Closure} is
  inspired by Example~\ref{ex:simple}. 
  Given a morphism $\omega : \Sigma^*\rightarrow \mathbb Z$, we compute a finite-state
  automaton $\mathcal A$ over an alphabet
  $\Gamma \subseteq (\Sigma^*)^4$ such that
  $\{ \varphi(xywz) :(x,y,w,z) \in L(\mathcal A)\}$ and
  $\varphi(L_R(\omega))$ have the same Zariski closure.  As in the case of
  the \textsc{Cover Closure} problem, Proposition~\ref{prop:eta} is
  instrumental in the construction of $\mathcal A$.  The automaton is
  based on the decomposition of words $w\in L_R(\omega)$ into three
  parts $w=w_1w_2w_3$ such that $w_1$ contains a factor $u$ with
  $\varphi(u)$ stable and $\omega(u)>0$ and $w_2$ contains a factor~$v$ with $\varphi(v)$ stable and $\omega(v)<0$ as shown in~\ref{fig:LZ-decompsoitionw}.  The details of the
  construction are much more involved than for the corresponding case of
  \textsc{Cover Closure}.  In particular, instead of abstractly identifying all
  sufficiently large counter values, as in the former case, we treat
  counter values exactly by encoding them in the input word to the
  automaton $\mathcal A$.

\subsection{Indexed Languages}
We show that for sets of matrices defined by extensions of
context-free languages, it is not possible to compute the Zariski
closure.

\begin{restatable}{theorem}{Indexed}
\label{theo:index-und}
There is no algorithm that given an indexed grammar~$G$ over an
alphabet~$\Sigma$ and a monoid morphism
$\varphi:\Sigma^* \to \mathrm{M}_d(\mathbb{Q})$, computes the
algebraic closure of~$\varphi(L(G))$.
\end{restatable}

The proof of undecidability is  by reduction from the boundedness problem
for reset VASS.  A reset VASS is a finite-state automaton equipped with
multiple counters taking values in non-negative integers.  The
counters can be incremented, decremented, or reset to zero.
A counter with value zero cannot be decremented, but there is no zero
test (a transition that only be taken when a counter is zero).
For this model, the boundedness problem (``is
the set of reachable configurations finite?'') is undecidable~\cite{DufourdFS98}.

The essence of our reduction is to take a reset VASS and produce an
indexed language $L\subseteq \Sigma^*$, each word of which encodes a
run of the VASS.  The encoding of a run is a complete list of all
prefixes of the sequence of transitions in the run.  We next construct
a morphism $\varphi$ with the property that $\Zcl{\varphi(L)}$ has
dimension at most one if and only if the reachable set of the VASS is
finite.  The morphism $\varphi$ maps each letter of $\Sigma$ to a
matrix that simulates the operation of incrementing and decrementing
the VASS counters.  The key challenge is to model the fact that a zero
counter cannot be decremented.  Our simulation is such that for any
word $w\in \Sigma$ that represents an illegal run of the VASS (in
which some counter takes a negative value) we have $\varphi(w)=\bm{0}$.
Thus only legal computations yield non-zero elements of $\varphi(L)$.
This is where the extra expressiveness of indexed grammars plays a key
role, specifically the fact that our encoding of VASS runs as words is
such that a word contains copy of each prefix of the run.

\section{Factorization Trees in Matrix Monoids}
\label{sec:factorisation}
Recall that $M \in M_d(\mathbb Q)$ is \emph{stable} if $\mathrm{rank}(M)=\mathrm{rank}(M^2)$.
A \emph{factorization tree} $\mathcal T$ in the monoid~$M_d(\mathbb Q)$ is 
  a finite unranked ordered tree: a node can have arbitrarily many children and the children of 
  each node are linearly ordered.  The nodes of $\mathcal T$ are labeled by elements of $M_d(\mathbb Q)$ such that the label 
  of a non-leaf node equals the left-to-right product
  of the labels of its children.  We require that the label $M$ of a node with three or more children be stable.   
  The \emph{height} of $\mathcal T$ is the maximum distance from the root to a leaf.  
  If $M_1,\ldots,M_m$ is a list of the labels of the leaf nodes, left-to-right, then we say that $\mathcal T$ is 
  a factorization tree of the sequence~$M_1,\ldots,M_m$. We  refer to this sequence as the \emph{yield} of the tree~$\mathcal T$.  

We call a sequence   $M_1,\ldots,M_m \in M_d(\mathbb Q)$ such that each matrix $M_i$ and their product $M_1\cdots M_m$ have rank $r$
a \emph{rank-$r$ sequence}. 
A consecutive subsequence $M_i,M_{i+1}\ldots,M_j$, with $1\leq i< j\leq m$,  of the sequence~$M_1,\ldots,M_m$ is called a \emph{segment}. 

\begin{proposition}\label{prop_factor._tree_height}
 All rank-$r$ sequences in $M_d(\mathbb Q)$
are the yield of some factorization tree of height at most~$d+2$.
\end{proposition}
\begin{proof}
Recall that for a vector space $V$ the \emph{exterior algebra} $\Lambda(V)$ is a vector space that contains a copy of~$V$ and is equipped with a
bilinear, associative, anti-symmetric operation $\wedge$.  
Below, we recall the key properties of~$\Lambda(V)$ needed to follow the proof, but see~\Cref{app:ext} for mode details.
If $V$ has finite dimension $d$ then $\Lambda(V)$ has dimension $2^d$.
The key properties of $\Lambda(V)$ for our purposes are that $v_1,\ldots,v_r \in V$ are linearly independent if and only if
 $v_1\wedge \ldots \wedge v_r\neq 0$, 
 and two such wedge products are nonzero scalar multiples of each other if and only if the corresponding sets of vectors span the same $r$-dimensional subspace. This defines a map $\iota$, which assigns  to each $r$-dimensional subspace $W$ of $V$ a vector  
 \[\iota(W):=v_1\wedge \ldots \wedge v_r\,,\] where $\{v_1,\ldots, v_r\}$ is an arbitrarily chosen basis of the subspace
 \footnote{
 The definition of~$\iota(W)$ depends up to a nonzero scalar multiple on the chosen  basis of~$W$.  
 }.
 Moreover, for subspaces $W_1,W_2$ of $V$, their intersection is trivial if and only if $\iota(W_1)\wedge\iota(W_2)\neq 0$.



The following claim is straightforward (see the proof in~\Cref{app-factorisation}):

\begin{restatable}{claim}{clmrankproduct}\label{clm:rank_product}
 Let  $M_1,\ldots,M_m \in M_d(\mathbb Q)$  be  a rank-$r$ sequence.   Then
$\rank(M_i \cdots M_k)=r$
for all $1\leq i< k \leq m$. Furthermore  $\im(M_{k+1}) \cap \ker(M_{k}) = \{0\}$  for $k\neq m$.
\end{restatable}

 Fix  a rank-$r$ sequence  $M_1,\ldots,M_m \in M_d(\mathbb Q)$.
We define the  greedy basis of the subspace of $\Lambda(\mathbb Q^d)$
spanned by $\{\iota(\im{M_{i}}) \mid 1\leq i < m\}$ to be the  
subsequence of indices
 $j_1<\cdots < j_s $ 
 such that $j_1=1$, and   
 $\iota(\im(M_{j_p}))$ is  not a linear combination of $\{\iota(\im{M_{i}}) \mid i < j_p \}$ for all $1 < p \leq s$. 
Note that  
\begin{equation}
    \label{eq-dim}
s \leq \dim \Lambda^r (\mathbb Q^d)\leq \dim \Lambda (\mathbb Q^d) = \sum_{r=0}^d\binom{d}{r}= 2^d\,.
\end{equation}

  Define $j_{s+1}:=m+1$.
  
\begin{claim}
\label{clm:stable_segment}
For every $\ell \in \{2, \ldots, s+1\}$, such that~$j_{\ell-1} \leq j_{\ell}-2$, there exists~$k \in \{1, \ldots, \ell-1\}$ such that the segment $M_{j_{k}} M_{j_{k}+1} \cdots M_{j_{\ell}-2}$ is stable.
\end{claim}
\begin{proofclaim}
\sloppy The condition $j_{\ell-1} \leq j_{\ell}-2$ implies that~$j_{\ell-1}< j_\ell -1<j_\ell$. 
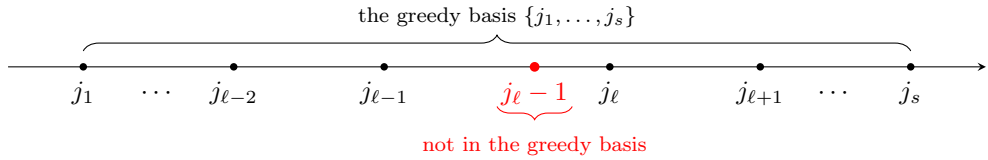
\begin{center}
     \begin{tikzpicture}[
    every node/.style={font=\small},
    dot/.style={circle, fill=black, inner sep=1pt},
    highlight/.style={circle, fill=red, inner sep=1.3pt},
    >=stealth
]

\draw[->] (0,0) -- (13,0);

\node[dot, label=below:{$j_{1}$}] (j1) at (1,0) {};
\node[dot, label=below:{$j_{\ell-2}$}] (jl2) at (3,0) {};
\node[dot, label=below:{$j_{\ell-1}$}] (jl1) at (5,0) {};
\node[highlight, label=below:{\textcolor{red}{$j_{\ell}-1$}}] (jlm2) at (7,0) {};
\node[dot, label=below:{$j_{\ell}$}] (jl) at (8,0) {};
\node[dot, label=below:{$j_{\ell+1}$}] (jlp1) at (10,0) {};
\node[dot, label=below:{$j_s$}] (js) at (12,0) {};

\node at (2, -0.35) {\small $\cdots$};
\node at (11, -0.35) {\small $\cdots$};
\draw[decorate,decoration={brace, amplitude=6pt}, yshift=6pt]
  (1,0) -- (12,0) node[midway, yshift=13pt] {\scriptsize the greedy basis $\{j_1,\dots,j_s\}$};

\draw[decorate, red, decoration={brace, mirror, amplitude=5pt}, yshift=-15pt]
  (6.5,0) -- (7.5,0) node[midway, yshift=-15pt] {\scriptsize \textcolor{red}{not in the greedy basis}};


\end{tikzpicture}

\end{center}

By construction of the greedy basis of ~$\{\iota(\im{M_{i}}) \mid 1\leq i < m\}$,  
for every index $i < j_{\ell}$, the vector~$\iota(\im(M_i))$ is a linear combination of $\iota(\im(M_{j_p}))$ where $j_p \in \{j_1, \dots, j_{\ell-1}\}$. 
Since $j_{\ell-1}< j_\ell -1<j_\ell$,
it follows that 
 there exists a nonzero vector~$(c_1,\ldots,c_{\ell-1})\in \mathbb Q^{\ell-1}$ where 
 \begin{equation}
 \label{eq:lc}
     \iota (\im(M_{j_\ell-1}))=\sum_{p=1}^{\ell-1} c_p\, \iota (\im(M_{j_p}))\,.
 \end{equation}
By Claim \ref{clm:rank_product},
$\im(M_{j_{\ell}-1})\cap \ker(M_{j_{\ell}-2})=\{0\}$, 
which implies  $\iota(\im(M_{j_{\ell}-1})) \wedge \iota(\ker(M_{j_{\ell}-2})) \neq 0$.
 Substituting the linear combination in~\eqref{eq:lc} and bilinearity of the wedge product yields
\[\sum_{p=1}^{\ell-1} c_p\, (\iota (\im(M_{j_p})) \wedge \iota(\ker(M_{j_\ell-2}))) \neq 0.\]
  Let   $k \in \{1, \ldots, \ell-1\}$  be an index such that $c_{k}(\iota (\im(M_{j_k})) \wedge \iota(\ker(M_{j_\ell-2}))) \neq 0$, giving that $\im(M_{j_k})\cap \ker(M_{j_\ell-2})=0$. Let $P=M_{j_k}\dots M_{j_\ell-2}$, and observe that $\rank(P)=r$ by Claim~\ref{clm:rank_product}.
  Since $\rank(P)=\rank(M_{j_k})$  we have~$\im(P)=\im(M_{j_k})$, and likewise since  $\rank(P)=\rank(M_{j_\ell-2})$   we have $\ker(P)=\ker(M_{j_\ell-2})$. As a result,~$\im(P)\cap \ker (P)=\{0\}$ and $P$ is stable. 
\end{proofclaim}

Now we are in a position to state the following claim:
\begin{restatable}{claim}{clmtreeheightnotfull}
     For $\ell \in \{2,\ldots,s+1\}$ there exists $k<\ell$ such that the segment $M_{j_k},\ldots,M_{j_{\ell-1}}$ has a factorization tree of height at most 2.
\end{restatable}
\begin{proofclaim}
     The claim holds if $j_{\ell-1}=j_{\ell}-1$ since then we take $k:=\ell-1$ and the segment $M_{j_k},\ldots,M_{j_{\ell-1}}$ is just a single matrix.
    In the case that  $j_{\ell-1} < j_{\ell}-1$,
     by  \Cref{clm:stable_segment}, there exists $k<\ell$ such that $P:=M_{j_k} \cdots M_{j_\ell-2}$ is stable. 
  It follows that
$P\cdot M_{j_{\ell}-1}$ has a factorization tree of height 2, as shown in~\Cref{fig:factorization}. 
\end{proofclaim}

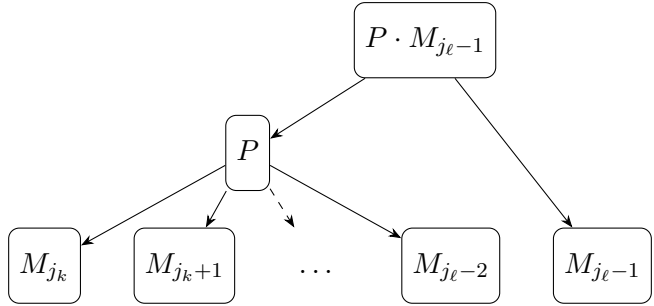
\begin{figure}[h!]
\centering
\begin{forest}
for tree={
  draw,
  rounded corners,
  minimum height=1cm,
  inner sep=4pt,
  l=1.5cm,
  s sep=7mm,
  edge={-Stealth},
  anchor=center,
  align=center
}
[{$P \cdot M_{j_{\ell}-1}$} 
  [{$P$} 
    [{$M_{j_k}$}, tier=leaf]
    [{$M_{j_k + 1}$}, tier=leaf]
    [{$\dots$}, draw=none, edge=dashed]
    [{$M_{j_{\ell}-2}$}, tier=leaf]
  ]
  [{$M_{j_{\ell}-1}$}, tier=leaf]
]
\end{forest}
\caption{factorization tree of $M_{j_k}, \dots, M_{j_{\ell}-1}$ with height at most 2}
\label{fig:factorization}
\end{figure}

Since each such segment $M_{j_k},\ldots,M_{j_{\ell-1}}$ has a factorization tree of height at most 2, and by~\eqref{eq-dim} there are at most $s \leq 2^d$ such segments, we can combine their trees using a binary tree of height at most~$\lceil \log_2 s \rceil\leq d$. Therefore, the final tree has height at most $d+2$, as desired.
\end{proof}

\proptree*

\begin{proof}

We construct the factorization tree in successive overlapping \emph{strata}, where each \emph{stratum} comprises at most three levels of the final tree
and the top level of a given stratum is the lowest level of the next stratum above.
The construction proceeds bottom-up and is such that the labels of all nodes in stratum $\ell$, except possibly the rightmost node in each level, have rank at most $d-\ell$.
The lowest level of the base stratum, $\ell = 0$, contains the leaf nodes corresponding to the input sequence of matrices (see~\Cref{fig:stratum}). 

We now describe how to construct a given stratum given list
$M_1, \ldots, M_m$ at its bottom level.
We group these nodes  to form products $N_1, \ldots, N_s$,
where each product $N_j$ is defined as
\begin{equation}\label{equation_stratum}
  N_j = M_{i_j} M_{i_j+1} \cdots M_{i_{j+1}-1}, \quad\quad  \text{for } j = 1, \ldots, s, 
\end{equation}
with indices $1 = i_1 < i_2 < \cdots < i_{s+1} = m + 1$ chosen so that the rank of each product~$N_j$ is equal to the rank of all its  factors:
\[\rank(N_j) = \rank(M_{i_j}) = \rank(M_{i_{j+1}})= \ldots = \rank(M_{i_{j + 1} -1}), \quad \text{for all } j.\]
%
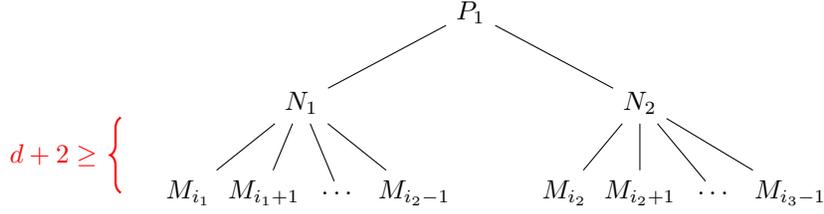
\begin{figure}[ht]
\centering
\begin{tikzpicture}[every node/.style={font=\small}, level distance=1.2cm]

\node (Mi1) at (0,0) {$M_{i_1}$};
\node (Mi1p1) at (1,0) {$M_{i_1+1}$};
\node (dots1) at (2,0) {$\dots$};
\node (Mi2m1) at (3,0) {$M_{i_2 -1}$};

\node (N1) at (1.5,1.2) {$N_1$};

\node (Mi2) at (5,0) {$M_{i_2}$};
\node (Mi3) at (6,0) {$M_{i_{2}+1}$};
\node (dots2) at (7,0) {$\dots$};
\node (Mi3m1) at (8,0) {$M_{i_3 -1}$};

\node (N2) at (6,1.2) {$N_2$};

\draw (Mi1) -- (N1);
\draw (Mi1p1) -- (N1);
\draw (dots1) -- (N1);
\draw (Mi2m1) -- (N1);

\draw (Mi2) -- (N2);
\draw (Mi3) -- (N2);
\draw (dots2) -- (N2);
\draw (Mi3m1) -- (N2);

\node (P1) at (3.75,2.4) {$P_1$};

\draw (N1) -- (P1);
\draw (N2) -- (P1);



\draw [decorate,decoration={brace,amplitude=4pt},red, thick]
  (-0.9,0) --++ (0,1) node[midway,xshift=-5pt,left,font=\small] {$ d+2 \geq $};

\node at (3.75,3.2) 
{\textbf{Fragment of Stratum 0}};
\end{tikzpicture}
\caption{Fragment of Stratum 0, showing two products \( N_1 \) and \( N_2 \) paired into a parent node \( P_1 \).}
\label{fig:stratum}
\end{figure}
We further require that the indices $i_1, \ldots, i_s$ are lexicographically maximal with the above property.
This means that the product $N_j$ includes as many consecutive $M_k$'s as possible, starting from $M_{i_j}$, subject to rank of the product being equal to $\rank(M_{i_j})$. 
It follows that for each $j \in \{1, \ldots, s-1\}$, the product $N_j N_{j+1}$ has strictly smaller rank than either $N_j$ or $N_{j+1}$. That is,
\[\rank(N_j N_{j+1}) < \max(\rank(N_j), \rank(N_{j +1})),\]
otherwise, we  would have $\rank(N_j) = \rank(N_{j +1}) = \rank(N_j N_{j+1})$ which contradicts  the lexicographic maximality of~$i_1, \ldots, i_s$.

\color{black}

 To build the third level, we combine the $N_j$'s in adjacent pairs: if $s$ is even, we form the products~$N_1 N_2,\, N_3 N_4,\, \ldots,\, N_{s-1} N_s$; and
if $s$ is odd, we form the products~$N_1 N_2$, $N_3 N_4,\, \ldots,\, N_{s-2} N_{s-1},\, N_s$.
Observe that, for all $j$,
\[\rank(N_j N_{j+1}) < \max(\rank(N_j), \rank(N_{j+1})) \le d - \ell \,.\]
 Since the top-most level in a stratum is the bottom-most level in next stratum, for all nodes $P$ in the next stratum, except possibly the rightmost one, we have
$\rank(P)  \le d - (\ell+1)$. Therefore, the rank bound is preserved in the new stratum, except possibly at the rightmost node (when $s$ is odd).

By Proposition \ref{prop_factor._tree_height}, each product $N_j$, as in equation~\eqref{equation_stratum}, can be computed using a factorization tree of height at most $d+2$, with leaves labeled $M_{i_j}, \ldots, M_{i_{j+1}-1}$.  Since there are at most $d$ strata, we obtain an overall bound of $d (d+3)$
for the total height of the factorization tree.
\end{proof}

We furthermore prove the following result, which is crucial in our study of  \textsc{Cover Closure} and \textsc{Reach Closure}. 
Let $\varphi:\Sigma^*\rightarrow M_d(\mathbb Q)$ and $\omega : \Sigma^*\rightarrow \mathbb Z$ be monoid morphisms,
where $\mathbb Z$ is considered additively and~$\omega(\sigma) \in \{-1,0,1\}$ for all $\sigma\in \Sigma$.

\decompose*

\begin{proof}
By symmetry, it suffices to prove the first item.  
We prove the contrapositive, that is, we assume for every factor $u$ of $w$ such that $\varphi(u)$ is stable we have $\omega(u) \le 0$, and argue that $\omega(w)< \eta(d)$. 
To this end, write $w=w_1w_2\cdots w_m$ and consider a factorization tree of the sequence of matrices
$\varphi(w_1),\ldots,\varphi(w_m)$.  By~\Cref{prop:tree} there is such a tree $\mathcal T$ of height at most $d(d+3)$.

For each node $\alpha$ of $\mathcal T$ we define corresponding string  $u_\alpha \in \Sigma^*$ bottom-up from the leaves.
For $\alpha$ the $i$-th leaf node we have $u_\alpha:=w_i$.  For a non-leaf node $\alpha$ with children $\alpha_1,\ldots,\alpha_\ell$,
we have $u_\alpha=u_{\alpha_1} \cdots u_{\alpha_\ell}$.  Clearly we have $u_\alpha=w$ for the root node $\alpha$.

We claim that for a node $\alpha$ of height~$h$ we have $\omega(u_\alpha)\leq 2^h$.   The proof is by induction on $h$.
The base case is that $\alpha$ is a leaf.  Then $u_\alpha=\sigma$ for some $\sigma\in\Sigma$, so $\omega(u_\alpha) \in \{-1,0,1\}$, and hence $\omega(u_\alpha)\leq 2^0=1$. 
For the inductive step, suppose $\alpha$ is a binary node with children $\alpha_1$ and $\alpha_2$ at height~$h$. 
Then 
\[\omega(u_\alpha)= \omega(u_{\alpha_1})+\omega(u_{\alpha_2}) \leq  2\cdot 2^{h-1}=2^h \, . \]
On the other hand, if $\alpha$ has three or more children then $\varphi(u_\alpha)$ (the label of $u_\alpha$) is stable by definition of a factorization tree and hence 
$\omega(u_\alpha)\leq 0$ by assumption.  Thus the induction hypothesis also holds in this case, and the claim is established.

Since the root node $\alpha$ has height at most $d(d+3)$ and satisfies $u_\alpha=w$,
 it follows from the above claim that 
\[   \omega(w)=\omega(u_\alpha) \leq 2^{d(d+3)}<\eta(d).\]
\end{proof}

\section{Algebraic Geometry Background}
\label{sec:background}

We recall basic notions from algebraic geometry, and  refer the reader
to~\cite{humphreys2012linear,borel2012linear} for an extended background.

An \emph{affine variety} (or \emph{algebraic set}) $X\subseteq  \overline{\Q}^d$ is the set of common zeros of finitely many polynomials  $P_1,\dots,P_\ell\in \overline{\Q}[x_1,\dots,x_d]$, that is,
\[X=\left\{x\in \overline{\Q}^d: P_1(x)=\dots=P_\ell(x)=0\right\}.\]
Given an ideal $I\subseteq  \overline{\Q}[x_1,\dots,x_d]$, the associated variety is 
$V(I)=\{x\in \overline{\Q}^d:\forall p\in I, P(x)=0\}$.
By Hilbert's Basis Theorem, every ideal $I$ is finitely generated. Affine varieties are precisely the sets of the form $V(I)$ for some ideal~$I\subseteq \overline{\Q}[x_1, \dots, x_d]$. 

The \emph{Zariski topology} on $\overline{\Q}^d$ takes varieties as closed sets. 
For $S\subseteq \overline{\Q}^d$, we denote by $\overline{S}$ its Zariski closure in the topology, that is, the smallest algebraic set containing~$S$.
Polynomial maps are continuous in Zariski topology, meaning that 
$f(\overline{S}) \subseteq \overline{f(S)}$ holds
for all  sets $S \subseteq \overline{\Q}^d$ and  polynomial maps $f: \overline{\mathbb{Q}}^d \to \overline{\mathbb{Q}}^{d'}$. The following fact is immediate, see a proof in~\Cref{app:background}.
%
\begin{restatable}{fact}{factclosureclosure}
    \label{fact_closure_closure}
    Let $f:X\to Y$ be a continuous function of topological spaces~$X$ and $Y$. Then $\overline{f(S)}=\overline{f(\overline{S})}$ for every subset~$S\subseteq X$.
\end{restatable}
The Zariski topology is \emph{Noetherian}: every descending chain of closed sets stabilizes.
A set is irreducible if it can not be written as the union of two proper closed subsets.
every closed set decomposes uniquely as a finite union of irreducible components—that is, maximal (respective to set inclusion)   irreducible closed subsets. The \emph{dimension} of a variety~$ X$ is the maximal length of a strict chain of irreducible closed subsets in $X$; when $X\subseteq \overline{\Q}^d$ we have $\dim(X)\leq d$.

In this paper, the two primary varieties of interest are   $M_d(\overline{\Q})\simeq \overline{\Q}^{d^2}$ and 
\[\GL_d(\overline{\Q})=\{(A, y) \in \overline{\Q}^{d^2 + 1} : \det(A) \cdot y = 1\},\]
viewed as an affine variety via the polynomial relation $\det(A)\cdot y=1$.
A \emph{linear algebraic group}~$G \subseteq \GL_d(\overline{\Q})$ is a group that is also an affine variety. The group operations (multiplication and inversion) are  polynomial maps— and hence are continuous with respect to the Zariski topology. 

The following is a classical result, see for example~\cite[1.2. Consequence]{boseck1993algebraic};
for completeness, we include a proof in~\Cref{app:background}.

\begin{restatable}{fact}{factclosedsemigroupissubgroup}
\label{fact:closed_semigroup_is_subgroup}
 A closed subsemigroup of $GL_d(\overline{\Q})$ is a linear algebraic group.   
\end{restatable}


\paragraph{Stable Matrices.}
The following lemma
on the Zariski closure of monoids generated by stable matrices is a variant of notions previously explored in~\cite[Proposition 11]{HrushovskiOP018}, adapted to our setting.
See~\Cref{app:background} for a detailed proof.

\begin{restatable}{lemma}{lemclosureisgroup}
 \label{lem:closure-is-group}
Let $M \in M_d(\mathbb Q)$ be  stable.  Then~$ \overline{\{M^n \, : \, n >0 \} \cap \{ N : \rank(N) = \rank(M)\}}$ is  a  group.
\end{restatable}





\section{Algorithm for Cover Closure }
\label{sec:cover}

In this section, we reduce  \textsc{Cover Closure} to the problem of computing 
 the Zariski closure of a finitely generated matrix monoid, for which there is a known algorithm~\cite{HrushovskiOPW23}.

\begin{proposition}
\label{prop:reduceNFA}
There is a reduction from  \textsc{Cover Closure}
 to the problem of computing the Zariski closure of a finitely generated matrix monoid.
\end{proposition}
\begin{proof}
Consider an instance of \textsc{Cover Closure}, given by monoid morphisms 
$\varphi:\Sigma^*\rightarrow M_d(\mathbb Q)$ and $\omega:\Sigma^*\rightarrow\mathbb Z$ with $\omega(\sigma)\in \{-1,0,1\}$ for all $\sigma\in \Sigma$.
The problem asks to compute 
the Zariski closure of $\varphi(L_C(\omega))$, where
$L_C(\omega)=\{ w\in \Sigma^* : \forall u \preceq w \cdot \omega(u) \geq 0\}$.  

We organize the proof around two key technical claims.  
At a high level, the first claim constructs an NFA $\mathcal{A}$ whose language $L(\mathcal A)$ is a superset of $L_C(\omega)$.
The second claim shows that this over-approximation is sufficiently tight such that $\varphi(L_C(\omega))$ and $\varphi(L(\mathcal A))$ have the same Zariski closure.

\begin{claim}
We can compute an NFA $\mathcal A$ such that 
\begin{enumerate}
\item $L_C(\omega) \subseteq  L(\mathcal A)$;
\item every word $w \in  L(\mathcal A) \setminus L_C(\omega)$ 
admits a factorization $w=w_1u w_2$, such that $\varphi(u)$ is stable,
$\omega(u)>0$, and~$\omega(w')\geq 0$ for every prefix $w'$ of $w_1u$.
\end{enumerate}
\label{clm:approxNFA}
\end{claim}

 \begin{proofclaim}
  The NFA $\mathcal A$ has set of states 
$Q:=\{0,1,\ldots, \eta(d)-1\} \cup \{\infty\}$ 
 for $\eta(d)=2^{d(d+3)}+1$, as defined in Proposition~\ref{prop:eta}.
State $0$ is initial and all states are accepting.
Intuitively,  the  states of $\mathcal A$ keep track the accumulated weight of the input word up to $\eta(d)-1$, thereafter treating the weight as infinite. 
The automaton blocks any symbol $\sigma$ with negative weight while in state $0$.  Reading a symbol with positive weight 
in state $\eta(d)-1$ leads to state $\infty$, which is a sink.
Formally,
the  transition relation $\Delta \subseteq Q\times \Sigma \times Q$ of $\mathcal A$ is given by
\begin{itemize}
    \item $(c,\sigma,c+\omega(\sigma)) \in \Delta$ if $c, c+\omega(\sigma) \in \{0,\ldots,\eta(d)-1\}$;
    \item $(\eta(d)-1,\sigma,\infty) \in \Delta$ if $\omega(\sigma)=1$;
    \item $(\infty,\sigma,\infty)\in \Delta$ for all $\sigma \in \Sigma$.
\end{itemize}

It is straightforward that $L_C(\omega) \subseteq L(\mathcal A)$, which is Item 1 of the claim. For Item 2, 
 let $w \in L(\mathcal A) \setminus L_C(\omega)$. 
 Then~$w$ is accepted along a run of $\mathcal A$ that enters the state $\infty$ at some point
 (otherwise the state of $\mathcal A$ exactly records the accumulated weight along the  word  $w$, ensuring that $w\in L_C(\omega)$).
For such a run, let $v$ be the prefix of $w$ leading to the first visit of state $\infty$.  Then 
 $\omega(v)=\eta(d)$.  By Proposition~\ref{prop:eta} the word $v$ admits a factorization~$v=w_1uw_2$ with $\varphi(u)$ stable and $\omega(u)>0$. 
 We furthermore have $\omega(w')\geq 0$ for every prefix $w'$ of $v$ by construction of $\mathcal A$.
 This yields the desired factorization of Item~2. 
 \end{proofclaim}

We next use Claim~\ref{clm:approxNFA} to show that $\varphi(L(\mathcal A))$ and $\varphi(L_C(\omega))$ have the same Zariski closure.

\begin{claim}
    If $\mathcal A$ satisfies Conditions 1 and 2 of~\Cref{clm:approxNFA} then
$\Zcl{\varphi(L_C(\omega))}=\Zcl{\varphi(L(\mathcal A))}$.
\label{clm:reduce}
\end{claim}

 \begin{proofclaim}
 Since $L_C(\omega) \subseteq L(\mathcal A)$,
 it suffices to show that for every word 
 $w \in L(\mathcal A)\setminus L_C(\omega)$ we have \[\varphi(w) \in \Zcl{\varphi(L_C(\omega))}\,.\]
 To this end, consider the factorization $w=w_1u w_2$ arising from Item~2 of~\Cref{clm:approxNFA}.  
 Since $\omega(u)>0$ and $\omega(w')\geq 0$ for all prefixes $w'$ of $w_1u$, there exists $n_0$ such that for all $n\geq n_0$ we have $w_1u^n w_2 \in L_C(\omega)$---specifically for $n$ sufficiently large it holds that $\omega(w')\geq 0$ for all prefixes $w'$ of $w_1u^nw_2$.

Since $\varphi(u)$ is stable, by Lemma~\ref{lem:closure-is-group}
$\Zcl{\langle \varphi(u)\rangle}$ is a group.  By~\Cref{fact:closed_semigroup_is_subgroup}, the subsemigroup  $\Zcl{\{\varphi(u^n) :  n\geq n_0\}}$  is the  group $\Zcl{\langle \varphi(u)\rangle}$.  Thereby contains a matrix $I_u$ such that $I_u\varphi (u)=\varphi (u)I_u=\varphi(u)$.
 The following multiplication map is Zariski continuous     \begin{align*}
    f:  M_d(\overline{\Q})  & \to M_d(\overline{\Q})\\
    A & \mapsto \varphi(w_1u) \,  A \, \varphi(w_2),
    \end{align*}
     Since $f(I_u) =  \varphi(w_1uw_2)=\varphi(w)$, we get 
    \begin{align}
    \label{eq:concov}
        \varphi(w) &\in f\left(\overline{\{\varphi(u)^n : n \ge n_0\}}\right) \nonumber \\
        & \subseteq  \, \overline{f\left(\{\varphi(u)^n : n \ge n_0\}\right)}= \Zcl{\{\varphi(w_1u^{1+n}w_2) :n\geq n_0\}} \\ 
        & \subseteq \overline{\varphi(L_C(\omega))}\, , \nonumber
    \end{align} 
   where the middle inclusion follows by continuity of~$f$.  
 \end{proofclaim}
This completes the reduction and its correctness proof.
\end{proof}

By~\cite[Theorem 3]{HrushovskiOPW23} there is a procedure to compute $\Zcl{\varphi(L(\mathcal A))}$ given an NFA $\mathcal A$ and 
morphism $\varphi:\Sigma^*\rightarrow M_d(\mathbb Q)$.  In combination with Proposition~\ref{prop:reduceNFA}, this gives  the decidability of the \textsc{Cover Closure} problem.

\section{Algorithm for Reach Closure}
\label{sec:reach}





In this section we give a procedure to decide \textsc{Reach Closure}.
There are two steps.  In the first step, we reduce \textsc{Reach Closure} to the problem \textsc{Zero Closure}, which we introduce below,
and in the second step we give a procedure for the problem \textsc{Zero Closure}.

\subsection{Reduction from Reach Closure to Zero Closure}
For a  morphism $\omega : \Sigma^* \rightarrow \mathbb Z$, define 
\[  L_{Z}(\omega) :=\{ w\in\Sigma^* : \omega(w)=0\} \]
to be the set of words of weight zero.  The \textsc{Zero Closure} problem asks, given monoid morphisms 
$\varphi:\Sigma^*\rightarrow M_d(\Q)$ and $\omega : \Sigma^* \rightarrow \mathbb Z$
 with $\omega(\sigma)\in \{-1,0,1\}$ for all $\sigma\in \Sigma$,  compute $\Zcl{\varphi(L_Z(\omega))}$.

\begin{proposition}
    \textsc{Reach Closure} reduces to \textsc{Zero Closure}.
\end{proposition}
\begin{proof}
Consider morphisms $\varphi: \Sigma^* \rightarrow M_d(\mathbb{Q})$ 
and $\omega:\Sigma^*\rightarrow \mathbb Z$.  Recall that the goal
of \textsc{Reach Closure} is to compute the closure of $\varphi(L_R(\omega))$. 
The following is the key technical claim.
\begin{claim}
We can compute a finite automaton $\mathcal A$ such that, writing
$L':=L(\mathcal{A})\cap L_Z(\omega)$,
we have $L_R(\omega) \subseteq L'$
and
every word $w \in  L'\setminus L_R(\omega)$ 
admits a factorization $w=w_1u w_2 v w_3$, where 
\begin{enumerate}
    \item $\varphi(u)$ is stable,   the weight of each  prefix  of $w_1u$ is nonnegative and $\omega(u)>0$;  
    \item $\varphi(v)$ is stable, the weight of each suffix of $vw_3$ is nonpositive and $\omega(v)<0$.  
\end{enumerate}
\label{clm:approxzvass} 
\end{claim}
\begin{proofclaim}
Automaton $\mathcal A$ has set of states 
\begin{gather}  
\label{eq:statesQZ}
Q:=(\{0,1,\ldots, \eta(d)-1\} \times \{0,1\}) \cup \{\infty\} \,  , 
\end{gather}
where $\eta(d)=2^{d(d+3)}+1$ is as defined in Proposition~\ref{prop:eta}.
Intuitively, the states of $\mathcal A$ keep track of the total weight of the word read so far, up to weight $\eta(d)-1$.
While words in the language
$L':=L(\mathcal{A})\cap L_Z(\omega)$ can have negative-weight prefixes, the construction of $\mathcal A$
ensures that any negative-weight prefix of a word  $w\in L'$ 
is sandwiched between the respective shortest and longest prefixes of $w$ of weight $\eta(d)$.
This is the key to proving the decomposition featured in the claim.

Formally, automaton $\mathcal A$
has  alphabet $\Sigma$, set of states $Q$ as shown in~\eqref{eq:statesQZ},
initial state $(0,0)$, and accepting states~$\{(0,0),(0,1)\}$.
The set $\Delta \subseteq Q\times \Sigma\times Q$ of transitions is defined such that
\begin{itemize}
    \item for all  $\sigma \in \Sigma$ and  $b\in\{0,1\}$,
    there is a transition $((c,b),\sigma,(c+\omega(\sigma),b))$  provided that
        $c,c+\omega(\sigma) \in \{0,\ldots,\eta(d)-1\}$;
           \item for all $\sigma \in \Sigma$
           there is a transition $(\infty,\sigma,\infty)$;
    \item for all $\sigma\in \Sigma$ with $\omega(\sigma)=1$ there is a transition
$((\eta(d)-1,0),\sigma,\infty)$;
\item for all $\sigma\in\Sigma$ with $\omega(\sigma)=-1$ there is transition $(\infty,\sigma,(\eta(d)-1,1))$.
\end{itemize}

 We first argue that  $L_R(\omega) \subseteq  L'$, that is,
$\mathcal A$ accepts any word  of zero weight all of whose prefixes have non-negative weight.  
Given a word $w\in L_R(\omega)$, suppose first 
that there is no prefix with weight $\eta(d)$.  Then there is a
 run of~$\mathcal A$ on $w$ that never enters the state~$\infty$ and exactly records the accumulated weight along the word~$w$. 
Otherwise,
let $x$ and $y$ be, respectively, the shortest and longest prefixes of $w$ with weight $\eta(d)$.
Then 
there is an accepting run of $\mathcal A$ on $w$ in which the state is $(c,0)$ after reading 
any proper prefix of $x$ of weight~$c<\eta(d)$; the state is $\infty$
after reading any prefix of $w$ containing $x$ and contained in $y$; the state is $(c,1)$ after reading any prefix of $w$ of weight~$c$ containing $y$.

We next prove Items 1 and 2 of the claim.
Consider a word $w \in L' \setminus L_R(\omega)$ and an accepting run of $\mathcal A$ on $w$. 
 Since $w\notin L_R(\omega)$ this run ends in $(0,1)$. Hence, $w$ has a shortest prefix $x$ with weight~$\eta(d)$.
Applying Item~1 of Proposition~\ref{prop:eta}, we obtain a decomposition $x=x_1u x_2$ where $\varphi(u)$ is stable and $\omega(u)>0$.
We moreover have that~$\omega(w')\geq 0$ for any prefix $w'$ of $x$ since the
construction of $\mathcal A$ prevents the weight of such prefix being negative.  Let $y$ be the shortest suffix of $w$ with weight~$-\eta(d)$.
Applying Item~2 of Proposition~\ref{prop:eta}, we have a decomposition~$y=y_1 vy _2$ where $\varphi(v)$ is stable and $\omega(v)<0$.
In addition,  $\omega(w')\leq 0$ for any suffix $w'$ of $y$ by construction of~$\mathcal A$.
This concludes the proof of the claim.
\end{proofclaim}

The following  shows how~\Cref{clm:approxzvass} can be used to compute the Zariski closure.

\begin{claim}\label{clm:reduce2}
    If $\mathcal A$ satisfies Conditions 1--2 of~\Cref{clm:approxzvass} then
$\Zcl{\varphi(L_R(\omega))}=\Zcl{\varphi(L')}$.
\end{claim}
\begin{proofclaim}
Since $L_R(\omega)\subseteq L'$,
it suffices to show that for every word  $w \in L' \setminus L_R(\omega)$ we have $\varphi (w) \in \Zcl{\varphi(L_R(\omega))}$.
For each such word~$w$,
 consider the factorization 
 $w=w_1u w_2 v w_3$ arising from Claim~\ref{clm:approxzvass}. 

Since $\omega(w)=0$, 
$\omega(u)>0$,  and $\omega(v)< 0$, there are infinitely many $m,n \in \mathbb N$ such that
\[\omega(w_1 u^m w_2 v^n w_3)=0.\]
Moreover since 
every prefix of $w_1u$ has nonnegative weight and every suffix of $vw_3$ has nonpositive weight,
for~$m,n$ sufficiently large we have that $w_1 u^m w_2 v^n w_3 \in L_R(\omega)$. We may thus define an infinite semigroup 
\[S:=\{ (m,n)\in \mathbb N^2 : w_1 u^{m+1} w_2 v^{n+1} w_3 \in L_R(\omega) \}.\]
Since $\varphi(u)$ and $\varphi(v)$ are stable, the matrix 
\[M:=\begin{pmatrix}
    \varphi(u)^m & 0\\
    0& \varphi(v)^n
\end{pmatrix}\]
is also stable. 
By
\Cref{lem:closure-is-group},  the Zariski-closed sub-semigroup $\overline{\langle M \rangle}$ is a group.
Hence, 
\[\overline{\{(\varphi(u)^m,\varphi(v)^n) :(m,n)\in S \}} \subseteq M_d(\overline{\Q})\times M_d(\overline{\Q})\]
is also a group,
and contains a matrix $(I_u,I_v)$
where $I_u$ is such that $I_u\varphi(u)=\varphi(u)I_u=\varphi(u)$ and  $I_v$ is such that $I_v\varphi(v)=\varphi(v)I_v=\varphi(v)$.
The following multiplication map is Zariski continuous     \begin{align*}
    f:  M_d(\overline{\Q}) \times M_d(\overline{\Q})  & \to M_d(\overline{\Q})\\
    (A,B) & \mapsto \varphi(w_1 \, u) \,  A \,\varphi(w_2) \,  B\, \varphi(v \, w_3),
    \end{align*}
     Since $f(I_u,I_v) =  \varphi(w_1 \, u\, w_2\, v\, w_3)=\varphi(w)$, by a similar argument to~\eqref{eq:concov}, we get 
 \[\varphi(w)=\varphi(w_1uw_2vw_3) \in \Zcl{\{\varphi(w_1u^{m+1}w_2v^{n+1}w_3) :(m,n) \in S\}} \subseteq \Zcl{\varphi(L_R(\omega))}\,.\]
\end{proofclaim}

 Thus, the reduction and its correctness proof follows from \Cref{sec:VASStoMON}.
\end{proof}

\subsection{Computing the Zero-Weight Language}
We show in this section how to compute the Zariski
closure~$\Zcl{\varphi(L_{Z}(\omega))}$.  The key idea is to represent
$\Zcl{\varphi(L_{Z}(\omega))}$ as the image of the Zariski closure of a
regular set over a product alphabet.
The proof uses the following bounded variant of $L_Z(\omega)$:
 \[ L_{BZ}(\omega) := \{ w\in\Sigma^* : \omega(w)=0 \wedge \forall u\preceq w \, (-\eta(d) \leq \omega(u)\leq \eta(d))\} \, .\]
 This is the language of words of weight $0$ such that the weight of every prefix 
 lies in~$\{-\eta(d),\ldots,\eta(d)\}$.


\begin{proposition}
The \textsc{Zero Closure} problem is decidable, that is, there is a procedure that given morphisms $\omega : \Sigma^*\rightarrow \mathbb
Z$ and $\varphi:\Sigma^* \rightarrow M_d(\Q)$, computes $\Zcl{\varphi(L_{Z}(\omega))}$.
\end{proposition}
\begin{proof}
Consider the alphabet $\Gamma = {(\Sigma\cup\{\varepsilon\})^4}\setminus  \{\varepsilon\}^4$.
Define the map $\mathrm{prod}:\Gamma^* \rightarrow (\Sigma^*)^4$ by
  \[ \mathrm{prod}((a_1,b_1,c_1,d_1) \cdots (a_m,b_m,c_m,d_m)) :=
(a_1\cdots a_m,b_1\cdots b_m ,c_1 \cdots c_m ,d_1 \cdots d_m) \] 
and define $\mathrm{flat}:\Gamma^*\rightarrow \Sigma^*$ by 
  \[ \mathrm{flat}((a_1,b_1,c_1,d_1) \cdots (a_m,b_m,c_m,d_m)) :=
a_1\cdots a_m\,b_1\cdots b_m \, c_1 \cdots c_m \, d_1 \cdots d_m\, .\] 
For a word $W \in \Gamma^*$ such that all but one component of $\mathrm{prod}(W)$ is $\varepsilon$, 
we may safely identify $W$ with $\mathrm{prod}(W)$, since the latter uniquely determines the former.

  We construct  a finite-state automaton $\mathcal A$ over the
alphabet $\Gamma$ such that
\begin{gather}
\Zcl{\varphi(L_{Z}(\omega))} = \Zcl{\varphi(L_{BZ}(\omega) \cup
  \mathrm{flat}(L(\mathcal A))) } \, .
\label{eq:claim}
\end{gather}
 Since $L_{BZ}(\omega)$ is a regular language,  the closure of its  image under $\varphi$  is computable.  
\begin{claim}
 $\overline{\varphi(\mathrm{flat}(L(\mathcal A)))}$ is computable.
 \end{claim}
 \begin{proofclaim}
 Consider the map $\psi$ that sends a word $w =(w_1, w_2, w_3, w_4) \in L(\mathcal A)$ to the block diagonal matrix $\diag ( \varphi(w_1), \varphi(w_2), \varphi(w_3), \varphi(w_4))$. By regularity of $L(\mathcal{A})$, it follows that  $\overline{\psi(L(\mathcal A))}$  is computable.
 The computability of $\overline{\varphi(\mathrm{flat}(L(\mathcal A)))}$ follows from a simple observation---namely, the fact that,  $\overline{\varphi(\mathrm{flat}(L(\mathcal A)))}$ is the closure of the image of  $\overline{\psi(L(\mathcal A))}$ under the polynomial map  that multiplies the diagonal blocks.
  \end{proofclaim}

The set of states of $\mathcal{A}$ is $\{-2\eta(d),\ldots,2\eta(d) \}$.  The state $0$ is initial and is the unique
accepting state.  There is a transition from $q\in Q$ to $q'\in Q$ with label
$(a,b,c,d) \in \Gamma$ precisely when $q'=q+\omega(abcd)$.  By
construction, for all words $W\in L(\mathcal A)$,
after reading $W$ the state
 of $\mathcal A$ is $\omega(\mathrm{flat}(W))=0$.  
Thus $\mathrm{flat}(L(\mathcal A)) \subseteq L_{Z}(\omega)$, which
establishes the right-to-left inclusion in~\eqref{eq:claim}.

It remains to prove the left-to-right inclusion in~\eqref{eq:claim}.
Suppose that $w \in L_{Z}(\omega)$.
We show that if $w\notin L_{BZ}(\omega)$ then $\varphi(w) \in
\overline{\varphi(\mathrm{flat}(L(\mathcal A)))} $.
To this end, we will construct two
words~$W,U\in \Gamma^*$ such that
$WU^k \in L(\mathcal A)$ for all~$k\in \mathbb N$, and moreover
    \begin{equation}
    \label{eq-flataim}
        \varphi(w) \in \overline{\varphi(\{\mathrm{flat}(WU^k) \, :\, k\in \mathbb{N} \})} \,.
    \end{equation}

By the assumption that $w \not\in L_{BZ}(\omega)$ there is a prefix $x \preceq w$
such that $\omega(x) \in \{-\eta(d),\eta(d)\}$.  Let $x$ be the shortest
such prefix.  We suppose that $\omega(x)=\eta(d)$;  the argument for
the case $\omega(x)=-\eta(d)$ is symmetric.  Then $w$ admits a decomposition
$w=xw_1yw_2$, shown in~\Cref{fig:LZ-decompsoitionw}, where  
\begin{itemize}
    \item $xw_1y$ is the shortest prefix of $w$ containing $x$ with weight $0$,  
    \item $y$ is the shortest suffix of $xw_1y$ with weight $-\eta(d)$,  and 
    \item $\omega(w_1)=\omega(w_2)=0$.
\end{itemize}


\begin{figure}[t]
    \centering
    \begin{tikzpicture}[scale=1.2]
  
  \draw[dashed, thick] (0,1)node[left] {$\eta(d)$~~}  -- (7.5,1) ;
  \draw[dashed, thick] (0,0) node[left] {$0$~~~~} -- (7.5,0) ;
  \draw[dashed, thick] (0,-1) node[left] {$-\eta(d)$~~} -- (7.5,-1) ;

 \draw[thick, blue, smooth] 
    plot coordinates {
      (0,0)
      (.8, .5)
      (1.2, -.7)
      (1.8,1.2)      
      (2.2,1.8)      
      (2.8,1.4)      
      (3.4,2.0)      
      (4.0,0.5)      
      (4.6,0.0)      
      (5.2,-1.5)     
      (5.8,0.8)      
      (6.4,-1.8)     
      (7.0,0.6)      
      (7.5,0.0)      
    };

 \draw[dashed, gray] (0,0) -- (0,-2.5);
\fill[black] (0,-2.5) circle (1.5pt);

\draw[dashed, gray] (1.72,1) -- (1.72,-2.5);
\fill[black] (1.72,-2.5) circle (1.5pt);

\draw[dashed, gray] (3.8,1) -- (3.8,-2.5);
\fill[black] (3.8,-2.5) circle (1.5pt);

\draw[dashed, gray] (4.6,0) -- (4.6,-2.5);
\fill[black] (4.6,-2.5) circle (1.5pt);

\draw[dashed, gray] (7.5,0) -- (7.5,-2.5);
\fill[black] (7.5,-2.5) circle (1.5pt);

\fill[blue] (0, 0) circle (2pt);
 \fill[blue] (1.72, 1) circle (2pt);
 \fill[blue] (3.8, 1) circle (2pt);
 \fill[blue] (4.6, 0) circle (2pt);
 \fill[blue] (7.5, 0) circle (2pt); 

 \draw[<->, thick] (0.1,-2.5) -- (1.62,-2.5) node[midway, above] {$x$};

 \draw[<->, thick] (1.82,-2.5) -- (3.7,-2.5) node[midway, above] {$w_1$};

  \draw[<->, thick] (3.9,-2.5) -- (4.5,-2.5) node[midway, above] {$y$};

   \draw[<->, thick] (4.7,-2.5) -- (7.4,-2.5) node[midway, above] {$w_2$};
   
\end{tikzpicture}
    \caption{Decomposition $w=xw_1yw_2$, where $x$ is the shortest prefix of $w$ with weight in $\{-\eta(d),\eta(d)\}$, the infix $xw_1y$ is the shortest weight-zero prefix of $w$ that extends $x$, and $y$ is the shortest suffix of $xw_1y$ with weight $\omega(y)=-\omega(x)$.}
    \label{fig:LZ-decompsoitionw}
\end{figure}
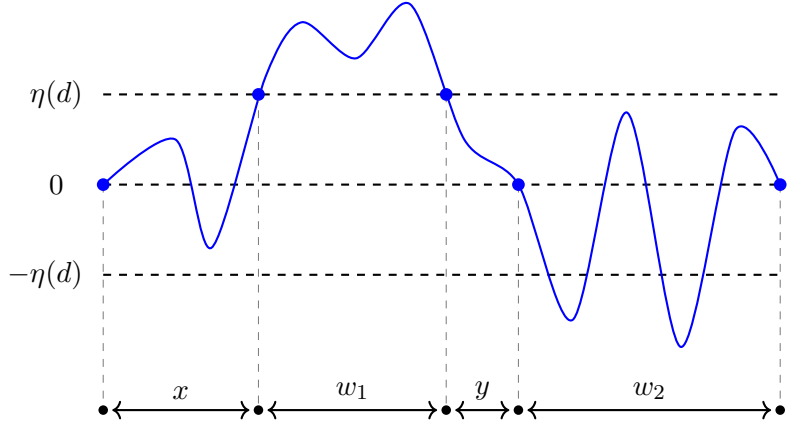

Applying Proposition~\ref{prop:eta} we have decompositions
$x=x_1ux_2$ and $y=y_2vy_1$ such that $\varphi(u)$ and $\varphi(v)$ are
stable, $\omega(u)>0$, and $\omega(v)<0$.

In summary, we have
\[ w = \underbrace{x_1\,  u \, x_2}_{\eta(d)}  \, \underbrace{w_1}_0
  \, \underbrace{y_2 \, v \, y_1}_{-\eta(d)} \, \underbrace{ w_2}_0 \]
with $\varphi(u),\varphi(v)$ stable, $\omega(u)>0$, and $\omega(v)<0$. 
From this decomposition of $w$ we construct a word $W \in L(\mathcal A)$ such
that
\begin{gather}
  \mathrm{flat}(W)=x_1\, u^{1+m}\, x_2\, w_1\, y_2\, v^{1+n}\,  y_1\,
  w_2 \, ,\label{eq:spec}
  \end{gather}
where $m,n\in \mathbb N$.
  We construct $W=W_1W_2W_3$ as the concatenation of three words $W_1,W_2,W_3 \in \Gamma^*$
such that 
\begin{enumerate}
\item $\mathrm{prod}(W_1) = (x_1u, x_2, y_2v, y_1)$;
\item  $\mathrm{prod}(W_2) = (u^{m_0} , w_1 , v^{n_0}, w_2)$;
\item $\mathrm{prod}(W_3) = (u^{m_1}, \varepsilon, v^{n_1}, \varepsilon)$
\end{enumerate}
for $m_0,m_1,n_0,n_1\in \mathbb N$.
Below we  specify the words $W_1,W_2,W_3$ and simultaneously describe 
an accepting run of~$\mathcal A$ over $W$.

We divide the run
into three phases, according to the factor that is being read.

\begin{itemize}
  \item
    \textbf{Phase 1.}
    In this phase we read $W_1$.
To unambiguously specify the word $W_1$ we stipulate that $W_1 = (x_1u,\varepsilon,\varepsilon,\varepsilon) \cdot (\varepsilon,x_2,\varepsilon,\varepsilon) \cdot
(\varepsilon,\varepsilon,y_2v,\varepsilon) \cdot (\varepsilon,\varepsilon,\varepsilon,y_1)$.
Since $x_1ux_2 y_2vy_1 \in L_{BZ}(\omega)$, the
automaton $\mathcal A$ can read $W_1$, starting and ending in state 0
with all intermediate states lying in $\{-\eta(d),\ldots,\eta(d)\}$; see~\Cref{fig:LZ-decompsoitionw}.

\item 
  \textbf{Phase 2.}  In this phase we read $W_2$.
  At the start of Phase 2 the automaton is in state $0$.
Suppose that~$w_1=\sigma_1\cdots \sigma_{\ell}$
and $w_2 = \xi_1\cdots\xi_k$.
The aim is to read
the sequence of symbols $(\varepsilon,\sigma_i,\varepsilon,\varepsilon)$ for
$i=1,\ldots,\ell$ followed by the sequence of symbols
$(\varepsilon,\varepsilon,\varepsilon,\xi_i)$ for $i=1,\ldots,k$.
At the same time, we aim to keep the 
automaton state in the bounded range
$\{-2\eta(d),\ldots,2\eta(d)\}$.  This can be done by interleaving the 
input with factors $(u,\varepsilon,\varepsilon,\varepsilon)$ and
$(\varepsilon,\varepsilon,v,\varepsilon)$ according to the following rule:
If ever the state equals $-\omega(u)$ then read  $(u,\varepsilon,\varepsilon,\varepsilon)$
to bring the state back to $0$.  Likewise,
if ever the state equals $-\omega(v)$, then read
$(\varepsilon,\varepsilon,v,\varepsilon)$ to bring the 
state back to $0$.  
See~\Cref{fig:Pahse2} for an illustration  of such an interleaving.
The numbers $m_0,n_0$ in the description of $W_2$
record the respective number of times we consume the factors 
$(u,\varepsilon,\varepsilon,\varepsilon)$
and $(\varepsilon,\varepsilon,v,\varepsilon)$ in Phase 2.
\item
  \textbf{Phase 3.}
  In this phase we read $W_3$.  
At the start of Phase 3 the state of the automaton $\mathcal A$ is in
the range~$\{-\eta(d),\ldots,\eta(d)\}$.
Moreover, since $\omega(w)=0$, we have 
\[\omega(\mathrm{flat}(W_1W_2)) = m_0 \omega(u)+n_0\omega(v)\,.\]
Now the equation $m\omega(u)+n\omega(v)=0$ has infinitely many solutions in positive integers~$m,n$.  Pick one such solution
$(m,n)$ such that $m_0\leq m$ and $n_0\leq n$. \\
Writing $m_1:=m-m_0$ and $n_1:=n-n_0$, our
aim in Phase~3 is to read $W_3$ such that 
\[\mathrm{prod}(W_3) = (u^{m_1},\varepsilon,v^{n_1},\varepsilon)\,.\]
This choice of $W_3$ ensures that $\omega(\mathrm{flat}(W_1W_2W_3))=0$.
We can read $W_3$ while staying in the range~$\{-2\eta(d),\ldots,2\eta(d)\}$ using the following greedy procedure.
Suppose that so far we have read $W_3' \in \Gamma^*$ such that
$\mathrm{prod}(W_3')= (u^{m'},\varepsilon,v^{n'},\varepsilon)$ with $m'\leq m_1$ and $n'\leq n_1$.  If $m'<m_1$  and
$n'<n_1$ then we read~$(u, \varepsilon, \varepsilon, \varepsilon)$ if
the current state is negative and we read
$(\varepsilon, \varepsilon, v,\varepsilon)$ if the current state is
positive.  If~$m'<m_1$ and $n'=n_1$ then we read
$(u, \varepsilon, \varepsilon, \varepsilon)$.  If $m'=m_1$ and $n'<n_1$
then we read $(\varepsilon, \varepsilon, v,\varepsilon)$. 
\end{itemize}

\begin{figure}[t]
    \centering
    \begin{tikzpicture}[scale=1.2]
  
  \draw[dashed, thick] (0,2) node[left] {$2\eta(d)$~~} -- (9.5,2) ;
  \draw[dashed, thick] (0,0) node[left] {$0$~~~~}  -- (9.5,0) ;
  \draw[dashed, thick] (0,-2) node[left]{$-2\eta(d)$~~} -- (9.5,-2);

    \draw[dashed, thick] (0,.8) node[left] {$-\omega(v)$~~} -- (9.5,.8) ;
  \draw[dashed, thick] (0,-.6) node[left] {$-\omega(u)$~~} -- (9.5,-.6) ;

  \draw[thick, blue, smooth] 
    plot coordinates {
      (0,0)
      (.4, .6)
      (.8, -.5) 
      (1.3,.8)
    };

\draw[line width=6pt, green!50, opacity=0.3, smooth] 
  plot coordinates { 
    (1.3,.8)
    (1.5, 1.2)
    (1.7,0)
  };

  \node[green!50!black, above] at (1.5, 1.3) {$v$};

    \draw[thick, green!70!black, smooth] 
    plot coordinates { 
      (1.3,.8)
      (1.5, 1.2)
      (1.7,0)
    };

    \draw[thick, blue, smooth] 
    plot coordinates {
      (1.7,0)
      (2.2,.5)
      (3,-.6)
    };

    \draw[thick, magenta, smooth] 
    plot coordinates {
      (3,-.6)
      (3.5,-.4)
      (4,-1.5)
      (4.3,0)
    };

    \draw[line width=6pt, magenta!50, opacity=0.3, smooth] 
  plot coordinates {
    (3,-.6)
    (3.5,-.4)
    (4,-1.5)
    (4.3,0)
  };
  
  \node[magenta, above] at (3.7, -1.6) {$u$};

  \draw[thick, blue, dotted, smooth] 
    plot coordinates {
      (4.3,0)
      (6, .6)
    };

  \draw[thick, blue, smooth] 
    plot coordinates {
      (6, .6)
      (6.2,.4)
      (6.5,.7)
       (6.9,-.6)
    }; 

     \draw[thick, magenta, smooth] 
    plot coordinates {
      (6.9,-.6)
      (7.4,-.4)
      (7.9,-1.5)
      (8.2,0)
    };

      \draw[line width=6pt, magenta!50, opacity=0.3, smooth] 
  plot coordinates {
    (6.9,-.6)
      (7.4,-.4)
      (7.9,-1.5)
      (8.2,0)
  };
  
  \node[magenta, above] at (7.6, -1.6) {$u$};

      \draw[thick, blue, smooth] 
    plot coordinates {
      (8.2, 0)
      (8.5,-.2)
      (9.2,.4)
    };

    \fill[green!70!black] (1.3,.8) circle (2pt);
    \fill[green!70!black] (1.7,0) circle (2pt);

     \fill[magenta] (3,-.6) circle (2pt);
    \fill[magenta] (4.3,0) circle (2pt);
    \fill[magenta] (6.9,-.6) circle (2pt);
    \fill[magenta] (8.2,0) circle (2pt);

    \fill[blue] (0,0) circle (2pt);
    \fill[blue] (9.2,.4) circle (2pt) node[right]{$m_0 \, \omega(u)+n_0\, \omega(n)$};

  \end{tikzpicture}
    \caption{
    Depiction of the word $W_2$, where the vertical axis represents the prefix weight.
    The parts in blue represent symbols $(\varepsilon,\sigma_i,\varepsilon,\varepsilon)$ or
    $(\varepsilon,\varepsilon,\varepsilon,\xi_i)$.  These are interleaved 
    with words $(u,\varepsilon,\varepsilon,\varepsilon)$ and $(\varepsilon,\varepsilon,v,\varepsilon)$
     so as to 
    keep the weight in the range $\{-2\eta(d),\ldots,2\eta(d)\}$.}
    \label{fig:Pahse2}
\end{figure}
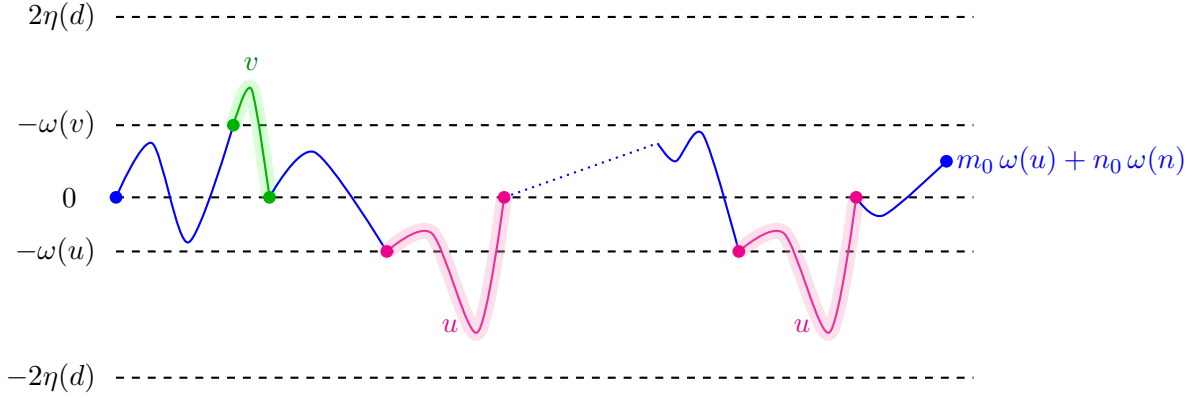

This completes the construction  of a word $W\in L(\mathcal A)$ such that there exist $m,n\in \mathbb{N}$ satisfying~\eqref{eq:spec}.
 Next we claim that
there exists $U \in L(\mathcal A)$ with
\[\mathrm{prod}(U)= (u^{m},\varepsilon,v^{n},\varepsilon)\,.\]

We construct $U$ and
an accepting run of $\mathcal A$ over $U$ in a greedy fashion exactly as described in Phase~3
above. This yields a run that stays in the range $\{-2\eta(d),\ldots,2\eta(d)\}$.
Since the unique accepting state of $\mathcal A$ is also the initial state, the language  $L(\mathcal A)$ is a semigroup, implying that $U^k \in L(\mathcal A)$
for all $k \in \mathbb N$.

In summary we have constructed words $W,U\in \Gamma^*$ such that
$WU^k \in L(\mathcal A)$ for all $k\in \mathbb N$.
To conclude the proof, it remains to prove~\eqref{eq-flataim}. Observe that
\[ \mathrm{flat}(WU^k) 
=x_1\, u^{1+(k+1)m}\, x_2\, w_1\, y_2\, v^{1+(k+1)n}\,  y_1\,
  w_2 \, , \]
  holds for all $k\in \mathbb N$, which in turn implies that 
\[ \varphi(x_1\, u^{1+(k+1)m}\, x_2\, w_1\, y_2\, v^{1+(k+1)n}\,  y_1\,
  w_2) \in \varphi(\mathrm{flat}(L(\mathcal A))). \]

Now $S=\{ (\varphi(u)^{(k+1)m},\varphi(v)^{(k+1)n}) : k \in \mathbb N\}$ is a
sub-semigroup of $M_d(\Q) \times M_d(\Q)$.  Since $\varphi(u)$ and~$\varphi(v)$ are stable, by Lemma~\ref{lem:closure-is-group}
$\Zcl{S} \subseteq M_d(\Zcl{\Q})\times M_d(\Zcl{\Q})$ is a group with respect to the operation of matrix
multiplication on each component.   In particular, $\Zcl{S}$ contains
an element $(I_u,I_v)$ such that $I_u \,\varphi(u) = \varphi(u)\, I_u 
=\varphi(u)$ and~$I_v \, \varphi(v) = \varphi(v)\, I_v 
=\varphi(v)$.  By Zariski continuity of matrix multiplication, and by a similar argument to~\eqref{eq:concov},  it follows that
\[ \varphi(w) = \varphi(x_1 u)\, I_u\, \varphi(x_2w_1y_2v)\, I_v\,
  \varphi(y_1) \in \Zcl{ \varphi(\mathrm{flat}(L(\mathcal A)))} \, .\]
This completes the proof of the left-to-right inclusion in~\eqref{eq:claim}.
\end{proof}

\section{Indexed Grammars}
\label{sec:indexG}
An indexed grammar $G$ is a tuple $(V,\Sigma,I,P,S)$ where
$V$ is the set of variables,~$\Sigma$ the set of terminals (or the alphabet), $I$ the set of indices, $S\in V$ the starting symbol, and $P$ is a finite set of production rules of the forms $A \to \alpha$, $ A \to B \, f$ and $A \, f \to \alpha$  
where $A,B \in V$, $f\in I$ and $\alpha$ is a word over $V \cup \Sigma$.
Derivations in an indexed grammar are similar to those in a context-free grammar except that the nonterminals can be followed by a string of indices.   
When a production rule $A \to BC$ is applied, the string of indices for $A$ is attached to both $B$ and $C$. See~\cite[Page 389]{hopcroft1979introduction} for more details; in the following we reiterate~\cite[Example 14.2]{hopcroft1979introduction} for the most familiar index language $\{a^nb^nc^n \mid n\geq 1\}$. We highlight that this language can be recognized as the reachability language of the following 2-VASS:

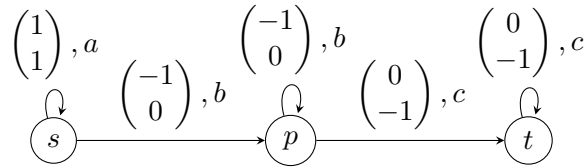
\begin{center}
       \begin{tikzpicture}[->,>=stealth,auto,node distance=2.5cm]
  \node (s) [circle,draw] {$s$};
  \node (ss) [circle,draw,right =2.5cm of s] {$p$};
  \node (t) [circle,draw,right =2.5cm of ss] {$t$};

  \path (s) edge[loop above] node[above] {$\begin{pmatrix} 1 \\ 1 \end{pmatrix},a$} (s)
            edge  node[above] {$\begin{pmatrix} -1 \\ 0 \end{pmatrix},b$} (ss)
        (ss) edge[loop above] node[above] {$\begin{pmatrix} -1 \\ 0 \end{pmatrix},b$} (ss)
       (ss) edge  node[above] {$\begin{pmatrix} 0 \\ -1 \end{pmatrix},c$} (t)
        (t) edge[loop above] node[above] {$\begin{pmatrix} 0 \\ -1 \end{pmatrix},c$} (t);;
\end{tikzpicture}
\end{center}

\begin{example}
    Let $G=(V,\Sigma,I,P,S)$ be such that $V=\{S,T,A,B,C\}$, $\Sigma=\{a,b,c\}$ and $I=\{f,g\}$
and with the following set $P$ of production rules
\[\begin{aligned}
    S &\to Tg  \qquad &\qquad   A\,f &\to a \,A \qquad & \qquad  A\,g & \to a \\
    T &\to T \,f \qquad &\qquad  B \, f &\to b \, B \qquad & \qquad  B\,g & \to b\\
    T &\to A\, B\, C \qquad &\qquad  C \, f &\to c \, C \qquad & \qquad  C\,g & \to c \\
\end{aligned}\]
An example derivation is 
\[\begin{aligned}
    S \Rightarrow  T\, g \,  & \Rightarrow T\, f\,g \, \Rightarrow A\, f\,g \, B\, f\,g \, C\, f\,g \, \Rightarrow a \,A  \,g \, B\, f\,g \, C\, f\,g 
     \, \Rightarrow a^2 \, B\, f\,g \, C\, f\,g \\
    & \Rightarrow a^2 \, b\, B  \,g \, C\, f\,g \, \Rightarrow a^2 \, b^2 \, C\, f\,g 
      \Rightarrow a^2 \, b^2\, c\, C\, g \, \Rightarrow a^2 \, b^2 \,c^2 \\
\end{aligned}\]
In the above derivation, if the production rule $ T \to T \,f $ at the second step repeats $n-1$ times the word $a^n \, b^n\, c^n$ is derived. \hfill $\blacktriangleleft$
\end{example}

Given an indexed grammar~$G$ 
over a finite alphabet~$\Sigma$ and  $\varphi:\Sigma^* \to \mathrm{M}_d(\mathbb{Q})$
a monoid homomorphism from $\Sigma^*$ to  $(d \times d)$-matrices with rational entries,
we are interested in the computability of the algebraic closure of~$\{\varphi(w) \mid w\in L(G)\}$.

\begin{theorem}
\label{theo:index-und}
    Given an indexed grammar~$G$ over an alphabet~$\Sigma$  and a monoid morphism $\varphi:\Sigma^* \to \mathrm{M}_d(\mathbb{Q})$, the algebraic closure of~$\varphi(L(G))$ is not computable. 
\end{theorem}

The proof is by reduction from boundedness problem for reset vector addition systems with states (reset VASSs), which is known to be undecidable~\cite{dufourd1998reset}.
Before we proceed with the proof, we give a brief introduction to reset VASSs and the boundedness problem for them.

\paragraph{Reset VASSs.}
A reset VASS $\mathcal{V}$ with dimension~$n \in \mathbb{N}$ is a tuple $(Q,q_0,\Delta)$ where $Q$ is
a finite set of states,~$q_0\in Q$ is the initial state and  $\Delta \subseteq Q \times  \{+1,0,-1,\textrm {reset}\}^n \times Q$ is the transition relation. 
A run of $\mathcal{V}$ is
a sequence $(q_0,\bm{x}_0), (q_1,\bm{x}_1), \cdots ,(q_k,\bm{x}_k)$   of state-vectors in $Q\times \mathbb{Z}^n_{\geq 0}$
such that for all $j\in \{1,\ldots,k\}$,
there exists~$(p,\bm{v},q) \in \Delta$
satisfying the following 
\begin{itemize}
    \item $q_{j-1}=p$, $q_j=q$, and $\bm{x}_0=\bm{0}$,
    \item the  $i$th coordinate of $\bm{x}_j$ is zero if the corresponding coordinate in $\bm{v}$ is $\mathrm{reset}$, and otherwise it is the summation of the $i$th coordinate of $\bm{v}$ and $\bm{x}_{j-1}$.
\end{itemize}

\sloppy A configuration  $(q,\bm{x})\in Q\times \mathbb{Z}^n_{\geq 0}$ is reachable  if for some $k\in \mathbb{N}$ there is a run~$(q_0,\bm{x}_0), (q_1,\bm{x}_2), \cdots, (q_k,\bm{x}_k)$ with $q_k=q$  and $\bm{x}_k=\bm{x}$. 
Denote by $\mathrm{Reach(\mathcal{V})}$ the set of reachable configurations of $\mathcal{V}$ in $Q\times \mathbb{Z}^n_{\geq 0}$.
The Boundedness problem  asks, given a reset VASS~$\mathcal{V}$, whether $\mathrm{Reach(\mathcal{V})}$ is finite.

\paragraph{Matrix Identities and Outer Product.}
Recall that the outer product of two column vector~$\bm{x},\bm{z} \in \mathbb{Q}^{n\times 1}$, denoted by $\bm{x} \otimes \bm{z}$, is defined by 
$(\bm{x} \otimes \bm{z})_{i,j}:=x_iy_j$ for all $1\leq i,j\leq n$.

Below, we denote by $\bm{0}\in \mathbb{Q}^{(n+1)\times 1} $ and $\bm{1} \in \mathbb{Q}^{(n+1)\times 1 }$  the all-zero and all-one column vectors, respectively. 
Denote by column vectors $\bm{e}_1,\ldots,\bm{e}_{n+1}\in \mathbb{Q}^{(n+1)\times 1}$ the standard basis.
Fix a  column vector~$\bm{y} = (y x_1, \dots, y x_n, y)^\top$  where $x_1,\ldots,x_n,y\in \mathbb{Q}$ .

\begin{fact}
\label{fact:outerprodcut}
The outer product $\bm{y} \otimes \bm{1}$ is a  matrix with identical columns equal to $\bm{y}$:
\begin{equation}
  \label{eq:yotimes1}  
\bm{y} \otimes \bm{1} = \bm{y} \bm{1}^\top =
\begin{pmatrix}
         yx_1 & yx_1 & \dots & yx_1  \\
        yx_2 & yx_2 & \dots & yx_2  \\
\vdots & \vdots & \ddots & \vdots  \\
yx_n& yx_n& \dots &yx_n\\
y & y   &\dots&y
    \end{pmatrix}\in \mathbb{R}^{(n+1)\times (n+1)}.
\end{equation} 
Similarly, $\bm 1\otimes \bm{y}$ is a  matrix with repeated rows~$\bm{y}^\top$.  
The product $\bm{1}\otimes \bm{1}$ is  a  matrix whose every entry is $1$. 
\end{fact}

The following claims are essential for the reduction:

\begin{claim}
\label{clm-updatedelta}
  Let $k\in\{1,\dots,n\}$. The product
$(\bm{y}\otimes \bm 1)(\bm 1\otimes (\bm e_k+\bm e_{n+1}))(\bm{y}\otimes \bm 1)$
is a  matrix whose columns are all equal to~$(n+1)y(x_k+1) \bm{y}$.  
\end{claim}
\begin{proofclaim}
    First we observe that the following product $L$ is a matrix where all columns are zero except the~$k$-th and $(n+1)$-th column being $(n+1)\bm{y}$:
\begin{align*}
  L= (\bm{y}\otimes \bm 1)(\bm 1\otimes(\bm e_k+\bm e_{n+1})) &= (\bm y\bm 1^\top) (\bm 1(\bm e_k+\bm e_{n+1})^\top)\\
  &= \bm y(\bm 1^\top \bm 1)(\bm e_k+\bm e_{n+1})^\top\\
   &= (n+1)\bm y  \otimes (\bm e_k+\bm e_{n+1})& \text{as } \bm{1}^\top \bm{1}=n+1 \, .\\
\end{align*}
Recall that all columns of  $R=(\bm{y}\otimes \bm 1)$ are  equal to~$\bm{y}$. Therefore

\begin{align*}
    LR&=(n+1)\begin{pmatrix}
         0 &\dots &0 & yx_1 & 0 &\dots& 0 &  yx_1  \\
       0 &\dots & 0 &yx_2 &0&\dots & 0 &  yx_2  \\
     \vdots &\ddots & \vdots & \vdots & \vdots &\vdots &\ddots &\vdots\\
     0 &\dots &0& yx_n & 0 &\dots& 0 &  yx_n \\
0 &\dots &0& y & 0 &\dots& 0 &  y 
    \end{pmatrix}\begin{pmatrix}
         yx_1 & yx_1 & \dots & yx_1  \\
        yx_2 & yx_2 & \dots & yx_2  \\
\vdots & \vdots & \ddots & \vdots  \\
yx_n& yx_n& \dots &yx_n\\
y & y   &\dots&y
    \end{pmatrix}\\
    &=(n+1)\begin{pmatrix}
         yx_1(yx_k+y) &  \dots & yx_1(yx_k+y) \\
         yx_2(yx_k+y) & \dots &  yx_2(yx_k+y) \\
\vdots &  \ddots & \vdots  \\
 yx_n(yx_k+y)& \dots & yx_n(yx_k+y)\\
 y(yx_k+y)& \dots & y(yx_k+y)
    \end{pmatrix}.
\end{align*}

We get that all columns of $LR$ are $(n+1)y(x_k+1) \bm{y}$.

\end{proofclaim}

Given a matrix $M$, we denote by $M_{i:}$ its $i$-th row. 
The proof of this second claim is  provided in~\Cref{app:indexG}.

\begin{restatable}{claim}{clmtarnsred}\label{clm:tarnsred}
  Let $M \in \mathbb{Q}^{(n+1)\times(n+1)}$. Then 
\[
{(M(\bm{y}\otimes \bm 1))}_{i:}=\begin{cases}
\bm{0}^\top & \text{if } M_{i:} = \bm{0}^\top, \\
(yx_i \pm y)\bm{1}^\top & \text{if } M_{i:} = (\bm{e}_i \pm \bm{e}_{n+1})^\top, \\
(yx_i)\bm{1}^\top & \text{if } M_{i:} = \bm{e}_i^\top.
\end{cases}\]   
\end{restatable}

Now we are in a position to formalize the reduction:

\begin{proof}[Proof of~\Cref{theo:index-und}]
Given a reset VASS $(Q,q_0,\Delta)$, we construct an index grammar $G$ over a finite alphabet~$\Sigma$ and  $\varphi:\Sigma^* \to \mathrm{M}_d(\mathbb{Q})$
a monoid homomorphism from $\Sigma^*$ to  $(d \times d)$-matrices with rational entries
such that 
\[\overline{ \varphi(L(G))}  \text{ has dimension at most one} \Longleftrightarrow \mathrm{Reach(\mathcal{V})} \text{ is finite.} \]

In our reduction, without loss of generality, we assume that $\mathcal{V}$ is such that,  for all transitions $(p,\bm{v},q)\in \Delta$, the vector $\bm{v}$ has at most one negative coordinate (as otherwise, introduction of more new states can convert~$\mathcal{V}$ to satisfy this assumption).
Roughly speaking, the reduction is such that for each sequence of transitions  $\delta_1,  \dots, \delta_k$ of $\mathcal V$ the grammar~$G$ produces a word over the alphabet~$\{\delta, t_{\delta} : \delta \in \Delta\} \cup \{\delta_0\}$ as follows 
\begin{equation}
\label{eq-derv2}
w= w_k \,  t_{\delta_k} \, w_k \, t_{\delta_{k-1}} \, w_{k-1} \, t_\delta{{k-2}} \, w_{k-2}  \cdots w_3 \,  t_{\delta_2} \, w_2 \, t_{\delta_1} w_1
\end{equation}
where $w_i$ is the suffix of $\delta_k \delta_{k-1}  \cdots \delta_0$ consisting of $(i+1)$-letter.
Observe that $w_i=\delta_i w_{i-1}=\delta_i\delta_{i-1}\cdots \delta_0$ for all~$i \in \{k,\ldots,2\}$.
For instance, given the sequence $\delta_1 \delta_2 \delta_3$ of transitions the grammar produces the following word
\begin{equation}
\label{eq-w-unde}
   \underbrace{\delta_3 \delta_2 \delta_1 \delta_0 \,  t_{\delta_3} \, \delta_3 \underbrace{ \delta_2 \delta_1 \delta_0 \, t_{\delta_2} \,  \delta_2 \underbrace{\delta_1 \delta_0 \, t_{\delta_1} \,  \delta_1 \delta_0}_{\mathrlap{\text{to check: can $\delta_1$ be taken at $(q_0,\bm{0})$?}}} }_{\mathrlap{\text{to check: can $\delta_1\delta_2$ be taken at $(q_0,\bm{0})$?}}} }_{\text{to check: can $\delta_1 \delta_2  \delta_3$ be taken at $(q_0,\bm{0})$?}} 
\end{equation}

The idea is that  if the sequence of transitions $\delta_1 \cdots \delta_k$ lifts to  an invalid run, then $\varphi(w)$ is the zero matrix, and otherwise it is an encoding of the reachable configuration by~$\delta_1   \cdots \delta_k$. 
Similar to the reduction in~\Cref{pro:VASStoMON},  
we encode the configurations $(q,\bm{x}) \in Q \times \mathbb Z^n_{\geq 0}$ of~$\mathcal{V}$
with matrices $M$ consisting of $|Q|\times |Q|$ blocks of $(n+1) \times (n+1)$  matrices.
The encoding is such that 
 all blocks of $M$ are  zero matrices except possibly one block at index $(q,q_0)$; this indicates that the run began in $q_0$ and reached $q$, with $\delta_k$ being the last transition (note that this is the reverse block compared to~\Cref{pro:VASStoMON}, where such a block has index~$(q_0, q$)).
Given a column vector~$\bm{y}$, by Fact~\ref{fact:outerprodcut}, the 
outer product~$\bm{y}\otimes \bm 1$
is a  matrix with repeated columns~$\bm y$. 
The vector~$\bm{x}=(x_1,\ldots,x_n)$ of the configuration~$(q,\bm{x})$
is encoded in the nonzero block of $M$ by the matrix 
$y (\bm{y}\otimes\bm{1})$ where     $\bm{y}=
(\bm{x},1)^\top=(x_1,\ldots,x_n,1)^\top$ is a column vector and $y$ is a nonzero scalar. 
We may refer to $y$ as the unit of the encoding. We consider the case when the unit is zero, meaning matrices with all blocks zero, encoding of an invalid configuration.

Since  the matrix multiplication is associative, we can study $\varphi(w)$ by considering the word backward, as shown in~\eqref{eq-w-unde}. The matrix product  $\varphi(\delta_1 \delta_0 t_{\delta_1} \delta_1 \delta_0)$ aims  at checking whether $\delta_1$ lift to a valid run starting at~$(q_0,\bm{0})$.
As briefly explained above, the main idea is that if taking $\delta_1$ is not allowed at the beginning of the run then~$\varphi(\delta_1 \delta_0 t_{\delta_1} \delta_1 \delta_0)=\bm{0}$. Otherwise, the construction is such that  \[\varphi(\delta_1 \delta_0 t_{\delta_1} \delta_1 \delta_0) = c \, \varphi(\delta_1 \delta_0)\,\]
for some nonzero $c\in \mathbb{Q}$, which implies that $\varphi(\delta_1 \delta_0 t_{\delta_1} \delta_1 \delta_0)$ is a valid encoding of the configuration reached after $\delta_1$.
Inductively, given that  $\delta_1 \ldots \delta_{i-1}$ lifts to a valid run in~$\mathcal V$, 
the image of suffix  $w_{i-1} t_{\delta_{i-1}}  w_{i-2} \cdots t_{\delta_1} w_1$ of $w$ is a multiple of image of $w_{i-1}$:
\[\varphi(w_{i-1} t_{\delta_{i-1}}  w_{i-2} \cdots t_{\delta_1} w_1) = c' \, \varphi(w_{i-1})\,\]
for some  nonzero $c'\in \mathbb{Q}$.
The construction next checks whether   $\delta_i$ can be taken in $\mathcal V$ after $\delta_1 \ldots \delta_{i-1}$.
This task is performed by the letter $t_{\delta_{i}}$ in the the infix $w_i t_{\delta_{i}}$ right before the suffix $w_{i-1} t_{\delta_{i-1}}  w_{i-2} \cdots t_{\delta_1} w_1$.  
By the invariant, on the right of $\varphi(t_{\delta_{i}})$
we have $c \, \varphi(\delta_i \delta_{i-1} \cdots \delta_0)$ and on the left we have $\varphi(\delta_i \delta_{i-1} \cdots \delta_0)$, both  encoding of the configuration reached after taking $\delta_1  \ldots \delta_i$, but possibly with  different units. 
These two copies enable the simulation of a quadratic update to the vector values and allow detection of whether any coordinate becomes  $-1$ as a result of~$\delta_i$. 
For instance, 
let $\bm{x}=(x_1,\ldots,x_n)$ be the vector reached after 
taking $\delta_0 \delta_1 \ldots \delta_i$,  the two encoding are in the form of 
$y_1 ((\bm{x},1)\otimes\bm{1})$ and $y_2((\bm{x},1)\otimes\bm{1})$ for some nonzero scalars~$y_1,y_2$. 
Let $\delta_i=(p,\bm{v},q)$ be such that the $\ell$-th coordinate of $\bm{v}$ is $-1$ (the only coordinate by assumption).
By~\Cref{clm-updatedelta},
the product
of these encoding on the left and right by 
$\bm 1\otimes (\bm e_\ell+\bm e_{n+1})$ results in 
a  matrix whose columns are all~$(n+1)y_1y_2(x_\ell+1)(\bm{x},1)$. 
If taking $\delta_i$ results the $x_\ell$ to become $-1$, this product is zero and gives the encoding of an invalid configuration. Otherwise, it is a valid encoding of~$\bm{x}$ with the unit~$(n+1) \, y_1 y_2 \, (x_\ell+1)$.

The grammar~$G=(V,\Sigma,I,P,S)$ is defined such that
\[V=\{S,A,B,C,D\} \qquad \Sigma=\{\delta,t_\delta\mid \delta\in \Delta\}\cup \{\delta_0\}\qquad I=\{f_\delta\mid \delta\in \Delta \}\cup \{f_{\delta_0}\}
\]
where $V$ is the set of variables,  $I$ the set of indices, and $S\in V$ the starting variable. 
We introduce the production rules in $P$ gradually, as it applies during a derivation, showing how it generates words of the form~\eqref{eq-derv2}. For all~$\delta\in \Delta$,
the rules 
\[S\to A f_{\delta_0} \qquad \qquad  A\to Af_\delta \]
generate derivations of the form $S\overset{*}{\Longrightarrow} A \, f_{\delta_k} \, f_{\delta_{k-1}}\dots \, f_{\delta_0}$. The rule $A\to DBC$ can be applied only once along a derivation,  generating $S\overset{*}{\Longrightarrow} Df_{\delta_k}\dots f_{\delta_0} Bf_{\delta_k}\dots f_{\delta_0} Cf_{\delta_k}\dots f_{\delta_0}$. 
Intuitively, the grammar is now committed to simulate the sequence  $\delta_1 \cdots \delta_k$ of transitions in~$\mathcal V$. 
Next, taking once $Cf_\delta\to BC$,  with $\delta\in \Delta$, generates 
 \[S\overset{*}{\Longrightarrow}Df_{\delta_k}\dots f_{\delta_0} Bf_{\delta_k}\dots f_{\delta_0}Bf_{\delta_{k-1}}\dots f_{\delta_0}Cf_{\delta_{k-1}}\dots f_{\delta_0}\,,\]
and repeating  this rule again gives
    \[
    S\overset{*}{\Longrightarrow}Df_{\delta_k}\dots f_{\delta_0} Bf_{\delta_k}\dots f_{\delta_0}Bf_{\delta_{k-1}}\dots f_{\delta_0}B f_{\delta_{k-2}}\dots Bf_{\delta_{1}}f_{\delta_0} Bf_{\delta_0}Cf_{\delta_0}.\]
For all $\delta \in \Delta$, the rules
\[Cf_{\delta_0}\to \epsilon \qquad \qquad Bf_\delta\to t_\delta \delta D\]
    combined together yields
    \[
    S\overset{*}{\Longrightarrow}Df_{\delta_k}\dots f_{\delta_0}\, t_{\delta_k} \delta_k Df_{\delta_{k-1}}\dots f_{\delta_0}\,t_{\delta_{k-1}} \delta_{k-1} Df_{\delta_{n-2}}\dots f_{\delta_0}\,\dots t_{\delta_2} \delta_2 Df_{\delta_{1}}f_{\delta_0} \,t_{\delta_1} \delta_1 Df_{\delta_0} \,.
    \]
Finally, the rules 
\[Df_\delta\to \delta D \qquad \qquad  D_{\delta_0}\to \delta_0 \]
  complete the generation of words of the form~\eqref{eq-derv2}.

\medskip

It remains to define the morphism~$\varphi: \Sigma^* \rightarrow M_d(\mathbb{Q})$ where $d=|Q| \times (n+1)$. 
As described above, each matrix~$\varphi(\sigma)$, for all~$\sigma \in \Sigma$,  consists of $|Q|\times |Q|$ blocks of $(n+1) \times (n+1)$  matrices,
with only one nonzero block. 
Define $\varphi(\delta_0)$ to be 
such a matrix with the nonzero block $B:=(\bm{0},1)\otimes \bm{1}$ 
at index $(q_0,q_0)$. Observe that $\varphi(\delta_0)$ encodes the initial configuration. 
Let $\delta\in \Delta$ such that $\delta=(p,\bm{v},q)$. 
The matrix 
$\varphi(\delta)$ has the nonzero block $B$ at index $(q,p)$ where $B$ is constructed as in~\Cref{clm:tarnsred}, where for all $i\in \{1,\ldots,n\}$, we define
\[B_{i:}=\begin{cases}
    \bm{0}^\top \quad & \text {if } v_i=\mathrm{reset}\\
   \bm e_i^\top \quad & \text{if } v_i=0\\
   ( \bm e_i\pm \bm e_{n+1})^\top \quad &\text{if } v_i=\pm 1 
\end{cases}\]
where $B_{i:}$ denotes the $i$-th row of $B$. We  define $B_{n+1:}$ to be $e_{n+1}$.

Recall that at most one coordinate of $\bm{v}$ is $-1$; let $\ell$ be the index of this coordinate if it exists, and~$\ell=1$ otherwise.
For the letter~$t_{\delta}$,  we define  $\varphi(t_{\delta})$ to be a matrix with the nonzero block $B$ at index $(q,q_0)$, where~$B:=\bm 1\otimes (\bm e_\ell+\bm e_{n+1})$ is chosen as in~\Cref{clm:tarnsred}.

The correctness proof  follows from the next claim, see the proof in~\Cref{app:indexG}:

\begin{restatable}{claim}{clindexgrammar}\label{cl:indexgrammar}
     The set $\mathrm{Reach(\mathcal{V})}$ is finite if, and only if, $\overline{ \varphi(L(G))}$   has dimension at most one.
\end{restatable}
\end{proof}

\paragraph{Acknowledgements}
Mahsa Shirmohammadi and Mahsa Naraghi were supported by the
ANR grant VeSyAM (ANR-22-CE48-0005).
James Worrell  and Rida Ait El Manssour were supported by EPSRC Fellowship EP/X033813/1. 
\bibliographystyle{plain}
\bibliography{ref}

\newpage
\appendix

\section{Detailed Computations for the Examples in the Overview}
\label{app:computation}

We start with \Cref{ex:anbndyck} where~$\varphi(a) = \begin{pmatrix}
    1 & 1\\
    0 & 1
\end{pmatrix}$ and $ \varphi(b) = \begin{pmatrix}
    1 & 0 \\
    1 & 1
\end{pmatrix}$.  Note that 
\[
\varphi(a)^m = \begin{pmatrix}
    1 & m\\
    0 & 1
\end{pmatrix}, \hspace{0.3cm} \varphi(b)^n = \begin{pmatrix}
    1 & 0 \\
    n & 1
\end{pmatrix},
\] 
then $\varphi(a)^m \varphi(b)^n = \begin{pmatrix}
    mn + 1 & m\\
    n & 1
\end{pmatrix}.$ Therefore, the Zariski closure of 
\[
\varphi(L_C) = \{ \varphi(a)^m \varphi(b)^n : m \ge n \} 
\] 
is the zero set of the ideal $I_C = \langle x_{11} - x_{12}x_{21} - 1, x_{22} - 1 \rangle.$ Moreover, the closure of 
\[
\varphi(L_R) = \{ \varphi(a)^n\varphi(b) ^n : n \in \N\}
\]
is the zero set of the ideal $ I_R = \langle  x_{11} - x_{12}x_{21} - 1, x_{22} - 1, x_{12} - x_{21} \rangle. $

\vspace{0.3cm}

For \Cref{ex:Dyck} where again $\varphi(a) = \begin{pmatrix}
    1 & 1\\
    0 & 1
\end{pmatrix}$ and $ \varphi(b) = \begin{pmatrix}
    1 & 0 \\
    1 & 1
\end{pmatrix}$, we have $\varphi(L)$ is stable under multiplication and contains $\{ \varphi(a)^n \varphi(b)^n : n \in \N\}.$ Therefore, 
\[
\overline{\{\varphi(a)^{n_1} \varphi(b)^{n_1} \varphi(a)^{n_2} \varphi(b)^{n_2}\varphi(a)^{n_3}\varphi(b)^{n_3} : n_1, n_2, n_3 \in \N\}} \subseteq \overline{\varphi(L)}.
\]
We will also show that 
\[
\overline{\varphi(L)} \subseteq \overline{\{\varphi(a)^{n_1} \varphi(b)^{n_1} \varphi(a)^{n_2} \varphi(b)^{n_2}\varphi(a)^{n_3}\varphi(b)^{n_3} : n_1, n_2, n_3 \in \N\}}.
\]
 In fact, the entries of $\varphi(a)^{n_1} \varphi(b)^{n_1} \varphi(a)^{n_2} \varphi(b)^{n_2}\varphi(a)^{n_3}\varphi(b)^{n_3}$ are polynomials in $n_1, n_2, n_3$. 
 Denote by~$M_{ij}(n_1, n_2 , n_3)$ the  entries of $\varphi(a)^{n_1} \varphi(b)^{n_1} \varphi(a)^{n_2} \varphi(b)^{n_2}\varphi(a)^{n_3}\varphi(b)^{n_3}$ for $1 \leq i,j \leq 2$.  We introduce new variables $t_1, t_2, t_3$ that  play the role of the integers $n_1, n_2, n_3$ and we consider the ideal $J \subseteq \mathbb Q[ x_{11}, x_{12},x_{21},x_{22}, t_1, t_2, t_3]$ defined by
 \[
 J \coloneqq \langle x_{11} - M_{11}(t_1, t_2, t_3), x_{12} - M_{12}(t_1, t_2, t_3), x_{21} - M_{21}(t_1, t_2, t_3), x_{22} - M_{22}(t_1, t_2, t_3)\rangle.
 \]
 By eliminating the variables $t_1, t_2, t_3$ with the help of a computer algebra system such as \textsc{Macaulay2} we find that these entries satisfy only one equation $x_{11} x_{22} - x_{12}x_{21} -1 = 0$.
Since the matrices $ \varphi(a)$ and $\varphi(b)$ have determinant 1, every element in $ \varphi(L)$ has determinant 1. We conclude that $\varphi(L)$ is the zero set of the ideal 
\[
I_L= \langle x_{11} x_{22} - x_{12} x_{21} -1 \rangle.
\]

\vspace{0.3cm}

We now move to \Cref{ex:RvsC2} where we considered two monoid morphisms $\varphi_1 : \Sigma^* \to M_2(\Q)$ and $\varphi_2 : \Sigma^* \to M_3(\Q)$. 
We start with $\varphi_1$  where $\varphi_1(a) = \begin{pmatrix}
    2 & 0 \\
    0 & 4 
\end{pmatrix}$ and $\varphi_1(b) = \begin{pmatrix}
    1 & 0 \\
    1 & 1
\end{pmatrix}.$ To compute the closure of the 1-VASS coverability and reachability we follow the reduction given in the proof of~\Cref{pro:VASStoMON}. We define the new morphism $\varphi_1' : \Sigma^* \to M_6(\Q)$ as defined in the proof of \Cref{pro:VASStoMON}. Therefore, we have

 $\varphi_1'(a) = \begin{pmatrix}
   \begin{pmatrix} \varphi_1(a) & 0\\0 &1 \end{pmatrix} & 0_{3\times3 }\\
   0_{3\times3 } & 0_{3\times3 }
\end{pmatrix}$ and $\varphi_1'(b) = \begin{pmatrix}
    0_{3\times3 } & \begin{pmatrix} \varphi_1(b) & 0\\0 &1 \end{pmatrix} \\
    \begin{pmatrix} \varphi_1(b) & 0\\0 &1 \end{pmatrix}  & 0_{3\times3 }
\end{pmatrix}.$

Now we solve the \textsc{Cover Closure} and the \textsc{Reach Closure} for the morphism $\varphi_1'$. First we define $\mathcal{M} \coloneqq \overline{ \langle \varphi_1'(a), \varphi_1'(b) \rangle}.$ We claim that the \textsc{Cover Closure} for $\varphi_1'$ reduces to computing $\mathcal{M}$. In fact, in the coverability language, every non-empty word must start with $a$. Moreover, for every~$w \in \Sigma^*$, since $\omega(a) = 1 > 0$ there exists an integer $n_0$ large enough that makes every prefix of $a^{n_0}w$ to have a non-negative total weight, which implies~$a^{n_0}w$ is in the coverability language. Then,   for every $n \ge n_0$, $a^n w$ is in the coverability language. Therefore,  for every $n \ge n_0$, $\varphi_1'(a)^n \varphi_1'(w)$ is in the closure of the image of the coverability  language by $\varphi_1'$. Since the matrix $\varphi_1'(a)$ is stable, that is, $\rank(\varphi_1'(a)) = \rank(\varphi_1'(a)^2)$  the semigroup $\overline{\langle \varphi_1'(a) \rangle}$ is a group and contains an identity element denoted by $I_a$ (See \Cref{lem:closure-is-group} for more details). Thus, by a similar argument as in the proof of \Cref{clm:reduce} which relies on the Zariski continuity  we have that $I_a \varphi(w)$ is in the closure of the image of the coverability language by $\varphi_1'$ for every $w \in \Sigma^*$. We conclude that  the closure of the image of the coverability  language by $\varphi_1'$ is $\overline{I_a \mathcal{M}} \cup \{ \id_6 \}$. 

For the \textsc{Reach Closure}, we  first define $\mathcal{N} \coloneqq \overline{\langle \varphi_1'(a)^n \varphi_1'(b)^n, \varphi_1'(b)^n \varphi_1'(a)^n : n \in \N \rangle }$. We claim that the  \textsc{Reach Closure} reduces to the computation of $\mathcal{N}$. In fact, every non-empty word in the reachability language must start with $a$ and end with $b$. Moreover, for every word $w \in \Sigma^*$ such that $\omega(w) = 0$, there exists $n_0$ large enough  that makes every prefix of $a^{n_0} w b^{n_0}$ have a non-negative total weight, which implies  $a^{n_0} w b^{n_0}$ is in the reachability language.  Then for every $n \ge n_0$, $a^{n_0} w b^{n_0}$ is in the reachability language. Therefore, $\varphi(a)^n \varphi(w) \varphi(b)^n$ is in the closure of the image of the reachability  language by $\varphi_1'$  for every $n \ge n_0$.
Since $ \varphi_1'(a), \varphi_1'(b)$ are stable matrices, the semigroup $\overline{\langle (\varphi_1'(a), \varphi_1'(b)) \rangle} \subseteq M_6(\overline{\Q}) \times M_6(\overline{\Q})  $ is a group and it contains an identity element $ ( I_a, I_b)$ (See \Cref{lem:closure-is-group}).
Then by the Zariski continuity and similar argument used in the proof of \Cref{clm:reduce2} we have that $I_a \varphi(w) I_b$ is in the closure of the image of the reachability  language by $\varphi_1'$ for every $w \in \Sigma^* $ with $\omega(w) =0$. Moreover, for every word $w \in \Sigma^*$ such that $\omega(w)$, we have $\varphi_1'(w) \in \mathcal{N}$ (this is also a consequence of the stability of the matrices $\varphi_1'(a), \varphi_1'(b)$ and the Zariski continuity). Thus we conclude that  the closure of the image of the reachability  language by $\varphi_1'$ is $\overline { I_a \mathcal{N} I_b} \cup \{ \id_6\}$.

Let us now compute $\mathcal{M}$,
note that
\[
\varphi_1(a)^m \varphi_1(b)^{n} = \begin{pmatrix}
    2^m& 0  \\
    n 4^m & 4^m 
\end{pmatrix},  \hspace{0.2cm} \varphi_1(b)^n \varphi_1(a)^{m} = \begin{pmatrix}
    2^m& 0  \\
    n 2^m & 4^m 
\end{pmatrix},
\]
then, define $M(x,y) \coloneqq \begin{pmatrix}
    x & 0 & 0 \\
    y & x^2 & 0\\
    0 & 0 & 1
\end{pmatrix}$ for every $x,y \in \overline{\Q}$. Hence,
\begin{align*}
\mathcal{M}  =  \Bigg\{ & \begin{pmatrix}
M(x,y)& 0_{3 \times 3}\\
    0_{3 \times 3} & 0_{3 \times 3}
\end{pmatrix}, \begin{pmatrix}
    0_{3 \times 3} & M(x,y)\\
    0_{3 \times 3} & 0_{3 \times 3}
\end{pmatrix} , \begin{pmatrix}
    0_{3 \times 3}& 0_{3 \times 3}\\
    M(x,y) & 0_{3 \times 3}
\end{pmatrix} , \begin{pmatrix}
   0_{3 \times 3}& 0_{3 \times 3}\\
   0_{3 \times 3} & M(x,y)
\end{pmatrix}, \\
& \begin{pmatrix}
  M(1,y)& 0_{3 \times 3}\\
   0_{3 \times 3} & M(1,y)
\end{pmatrix}, \begin{pmatrix}
   0_{3 \times 3}& M(1,y)\\
   M(1,y) & 0_{3 \times 3}
\end{pmatrix} : x,y\in \overline{\Q} \Bigg\}, 
\end{align*}
where the matrices $\begin{pmatrix}
  M(1,y)& 0_{3 \times 3}\\
   0_{3 \times 3} & M(1,y)
\end{pmatrix}, \begin{pmatrix}
   0_{3 \times 3}& M(1,y)\\
   M(1,y) & 0_{3 \times 3}
\end{pmatrix}$  correspond to $\overline{\langle \varphi'_1(b) \rangle}.$
 \begin{itemize}
     \item The  closure of the image of the coverability  language by $\varphi_1'$ is
     \[
     \overline{I_{a} \mathcal{M}} \cup \{ \id_6\}  = \left\{ \id_6, \begin{pmatrix}
M(x,y)& 0_{3 \times 3}\\
    0_{3 \times 3} & 0_{3 \times 3}
\end{pmatrix}, \begin{pmatrix}
    0_{3 \times 3} & M(x,y)\\
    0_{3 \times 3} & 0_{3 \times 3}
\end{pmatrix}  : x,y \in \overline{\Q} \right\}.
\]
         \item  Note that,  the Zariski closure $\overline{\langle \varphi_1(a)^n \varphi_1(b)^n : n \in \N \rangle }$ is $\left\{ \begin{pmatrix}
             x & 0\\
             y & x^2
         \end{pmatrix} : x, y \in \overline{\Q}\right\}$, which implies that in this example we have $\mathcal{M},$ and $\mathcal{N}$ coinside, that is,     
         \[
         \mathcal{N} \coloneqq \overline{\langle \varphi_1'(a)^n \varphi_1'(b)^n, \varphi_1'(b)^n \varphi_1'(a)^n : n \in \N\rangle} = \overline{\langle \varphi_1'(a), \varphi_1'(b) \rangle} = \mathcal{M}. 
         \]
          Thus,  the closure of the image of the reachability  language by $\varphi_1'$ is  \[
         \overline{ I_a\mathcal{M} I_b} \cup \{ \id_6\}   = \left\{\id_6, \begin{pmatrix}
M(x,y)& 0_{3 \times 3}\\
    0_{3 \times 3} & 0_{3 \times 3}
\end{pmatrix}, \begin{pmatrix}
    0_{3 \times 3} & M(x,y)\\
    0_{3 \times 3} & 0_{3 \times 3}
\end{pmatrix}  : x,y \in \overline{\Q} \right\}.
\]
 \end{itemize}

We go back to the original problem of computing the closure of the 1-VASS coverability and reachability for the morphism $\varphi_1$. As explained in the proof of \Cref{pro:VASStoMON} we conclude that the closures of  $ \varphi_1(L_C)$ and~$\varphi_1(L_R)$ coincide and are given as the zero set of the ideal $ \langle x_{12}, x_{11}^2 - x_{22} \rangle$.

Now we consider the morphism $\varphi_2$ where  
\[
\varphi_2(a) =\begin{pmatrix}
    1 & 1 & 0 \\
    0 & 1 & 1 \\
    0 & 0 & 1  
\end{pmatrix} , \hspace{0.3cm} \varphi_2(b) = \begin{pmatrix}
    1 & - 1 & 0 \\
    0 & 1 & 1 \\
    0 & 0 & 1 
\end{pmatrix}.\] We repeat the same steps as we have done  for the morphism $ \varphi_1'$.
We define again the new morphism $ \varphi_2' : \Sigma^* \to M_8(\Q)$ as in the proof of \Cref{pro:VASStoMON}:

 \[\varphi_2'(a) = \begin{pmatrix}
   \begin{pmatrix} \varphi_2(a) & 0\\0 &1 \end{pmatrix} & 0_{4\times4 }\\
   0_{4\times4 } & 0_{4\times4 }
\end{pmatrix},\hspace{0.3cm}\varphi_2'(b) = \begin{pmatrix}
    0_{4\times4 } & \begin{pmatrix} \varphi_2(b) & 0\\0 &1 \end{pmatrix} \\
    \begin{pmatrix} \varphi_2(b) & 0\\0 &1 \end{pmatrix}  & 0_{4\times4 }
\end{pmatrix}.\]
Note that: 
\[
\varphi_2(a)^m \varphi_2(b)^{n} = \begin{pmatrix}
    1 & m - n & 1/2(m^2 - n^2 + n - m + mn)\\
    0 & 1 & m + n \\
    0 & 0 & 1
\end{pmatrix}.
\]
In fact, the entry $x_{12}$ reflects the number of $\varphi_2(a)$ minus the number of $\varphi_2(b)$ in every product of $\varphi_2(a)$'s and~$\varphi_2(b)$'s. Thus, we have 
\small
\begin{align*}
    \mathcal{M} & \coloneqq \overline{\langle \varphi'_2(a), \varphi_2'(b) \rangle }\\
    & = \Bigg\{  \begin{pmatrix}
M(x,y,z)&  0_{4 \times 4}\\
     0_{4 \times 4} &  0_{4 \times 4}
\end{pmatrix}, \begin{pmatrix}
     0_{4 \times 4} & M(x,y,z)\\
     0_{4 \times 4} &  0_{4 \times 4}
\end{pmatrix} , \begin{pmatrix}
     0_{4 \times 4}&  0_{4 \times 4}\\
    M(x,y,z) &  0_{4 \times 4}
\end{pmatrix} , \begin{pmatrix}
    0_{4 \times 4}&  0_{4 \times 4}\\
    0_{4 \times 4} & M(x,y,z)
\end{pmatrix}, \\
& \hspace{1cm} \begin{pmatrix}
  M(-x,x,\tfrac{x(x-1)}{2})&  0_{4 \times 4}\\
    0_{4 \times 4} &M(-x,x,\tfrac{x(x-1)}{2})
\end{pmatrix}, \begin{pmatrix}
   0_{4 \times 4}& M(-x,x,\tfrac{x(x-1)}{2})\\
  M(-x,x,\tfrac{x(x-1)}{2}) &  0_{4 \times 4}
\end{pmatrix} : x,y,z\in \overline{\Q} \Bigg\},
\end{align*}
\normalsize
where
$M(x,y,z) \coloneqq  \begin{pmatrix}
    1 & x & z & 0\\
    0 & 1 & y & 0 \\
    0 & 0 & 1 & 0 \\
    0 & 0 & 0 & 1
\end{pmatrix}.$
Then  the closure of the image of the coverability  language by $\varphi_2'$ is 
     \begin{align}\label{eq:N}
\overline{I_{a} \mathcal{M}}\cup\{\id_8\}  = \Bigg\{ & \id_8, \begin{pmatrix}
M(x,y,z)&  0_{4 \times 4}\\
     0_{4 \times 4} &  0_{4 \times 4}
\end{pmatrix}, \begin{pmatrix}
     0_{4 \times 4} & M(x,y,z)\\
     0_{4 \times 4} &  0_{4 \times 4}
\end{pmatrix} , \begin{pmatrix}
     0_{4 \times 4}&  0_{4 \times 4}\\
    M(x,y,z) &  0_{4 \times 4}
\end{pmatrix} ,  \notag \\
& \begin{pmatrix}
    0_{4 \times 4}&  0_{4 \times 4}\\
    0_{4 \times 4} & M(x,y,z)
\end{pmatrix}  : x,y, z\in \overline{\Q} \bigg\}. 
\end{align}

In order to compute  the closure of the image of the reachability  language by $\varphi_2'$, we first need to compute 
\[
\mathcal{N} \coloneqq\overline{\langle \varphi'_2(a)^n \varphi_2'(b)^n,\varphi'_2(b)^n \varphi_2'(a)^n : n \in \N \rangle}.
\]
By computing $\varphi'_2(a)^n \varphi_2'(b)^n,\varphi'_2(b)^n \varphi_2'(a)^n$ for any $n \in \N$ we can realize that:
\small
\[
\mathcal{N} =   \left\{ \begin{pmatrix}
M(0,y,z)&  0_{4 \times 4}\\
     0_{4 \times 4} &  0_{4 \times 4}
\end{pmatrix}, \begin{pmatrix}
     0_{4 \times 4} & M(0,y,z)\\
     0_{4 \times 4} &  0_{4 \times 4}
\end{pmatrix} , \begin{pmatrix}
     0_{4 \times 4}&  0_{4 \times 4}\\
    M(0,y,z) &  0_{4 \times 4}
\end{pmatrix} , \begin{pmatrix}
    0_{4 \times 4}&  0_{4 \times 4}\\
    0_{4 \times 4} & M(0,y,z)
\end{pmatrix}  : y, z\in \overline{\Q} \right\}.
\]
\normalsize
Hence,  the closure of the image of the reachability language by $\varphi_2'$ is:
\begin{equation}\label{eq:NZ}
\overline{I_a \mathcal{N} I_b} \cup\{\id_8\} =  \left\{ \id_8, \begin{pmatrix}
M(0,y,z)&  0_{4 \times 4}\\
     0_{4 \times 4} &  0_{4 \times 4}
\end{pmatrix}, \begin{pmatrix}
     0_{4 \times 4} & M(0,y,z)\\
     0_{4 \times 4} &  0_{4 \times 4}
\end{pmatrix} : y,z \in \overline{\Q} \right\}.
\end{equation}
From \Cref{eq:NZ} and \Cref{eq:N} we conclude that the closures of $\varphi_2(L_C)$ and $\varphi_2(L_R)$ are given as the zero set of the ideals $I_C = \langle x_{11} - 1, x_{22} - 1, x_{33} - 1, x_{21},  x_{31},x_{32}\rangle$ and $I_R = I_C + \langle x_{12} \rangle$ respectively.

\section{Proof of Proposition~\ref{pro:VASStoMON}}
\label{sec:VASStoMON}
\VASStoMON*
\begin{proof}
We give the proof for the case of 1-VASS coverability languages.  The case of reachability languages follows with minor modifications.
Since the language $L_C(\omega)$ is the coverability language of a (one-state) 1-VASS, it suffices to reduce the problem of computing
the closure of a 1-VASS coverability language to \textsc{Cover Closure}.  To this end,
we first observe that every 1-VASS coverability language is the homomorphic image of a language of the form 
$L\cap L_{C}(\omega) \subseteq \Sigma^*$ for some finite alphabet $\Sigma$ and regular language $L$.  Indeed, given a 1-VASS $\mathcal V$, let $\Sigma$ be the set of transitions of $\mathcal V$,
$L$ the set of sequences of transitions from an initial to accepting state, and $\omega$ the map assigning each transition to its weight.  Then  the coverability language of $\mathcal V$ is the image 
of $L\cap L_{C}(\omega)$ under the morphism that maps each transition to its label.
To realize the claimed reduction, it thus suffices to reduce the problem of computing the Zariski closure of the set $\varphi(L\cap L_{C})$
for a morphism $\varphi:\Sigma^*\rightarrow M_d(\mathbb Q)$ to \textsc{Cover Closure}.  This is shown in the following claim:
\begin{claim}
\label{prop:reductionwithnostate}
    The problem of computing the Zariski closure of $\varphi(L_{C}(\omega)\cap L)$, given a regular language $L\subseteq \Sigma^*$ and morphisms $\varphi:\Sigma^*\rightarrow M_d(\mathbb Q)$ and $\omega:\Sigma^*\rightarrow\mathbb Z$, can be reduced to \textsc{Cover Closure} problem.
\end{claim}
\begin{proofclaim}
    Let $\mathcal A$ be a deterministic finite automaton that accepts language $L$.  Suppose that $A$ has set of states $Q=\{q_1,\ldots,q_k\}$ with $q_1$ initial and $F\subseteq Q$ the set of  accepting states.
    Define  
 a morphism $\varphi': \Sigma^* \rightarrow M_{k(d+1)}(\mathbb Q)$
where each $\varphi'(\sigma)$ is a block matrix of the form 
\[ \varphi'(\sigma) := \begin{pmatrix}  
  M_{1,1} & M_{1,2} & \cdots & M_{1,k}\\
  M_{2,1} & M_{2,2} & \cdots & M_{2,k} \\
   \vdots & \vdots & \ddots & \vdots \\
  M_{k,1} & M_{k,2} & \cdots & M_{k,k} 
\end{pmatrix} \,,\]
where blocks $M_{i,j}$ are zero except for those corresponding to transitions
$(q_i,\sigma,q_j)$ of~$\mathcal{A}$ labeled by~$\sigma$.  In such cases, for each transition from state $q_i$ to $q_j$,
we set $M_{i,j} = \begin{pmatrix} \varphi(\sigma) & 0\\0 &1 \end{pmatrix}$.\footnote{The bottom-right entry 1 of the matrices $M_{i,j}$ in $\varphi'(\sigma)$
allows us to distinguish between the presence of a transition $(q_i,\sigma,q_j)$ with $\varphi(\sigma)=\boldsymbol 0$ (modeled by 
$\begin{pmatrix} \boldsymbol 0&0\\0 & 1 \end{pmatrix}$) and the absence of a transition $(q_i,\sigma,q_j)$ (modeled by setting the whole block $M_{i,j}$ to be zero). Note that this corrects an oversight in~\cite[Section 4]{HumenbergerJK18}.}
Then the Zariski closure of $\varphi(L_{C}(\omega)\cap L)$
can be computed from that of 
$\varphi'(L_{C}(\omega))$ as follows.
First intersect
$\varphi'(L_{C}(\omega))$ with the (closed) set of matrices~$M$ such that for some state $q_i\in F$ the bottom-right entry of block $M_{1i}$ equals 1.  Then map the resulting set into~$M_d(\mathbb Q)$ by sending  matrix $M$ to $\sum_{q_i\in F} M_{1i}$
and projecting out the last row and column.
\end{proofclaim}
Having established the claim we have completed the proof of the proposition.
\end{proof}
\begin{corollary}
    The problems \textsc{Cover Closure} and \textsc{Reach Closure} can be reduced to the special case in which each alphabet symbol
    $\sigma\in \Sigma$ has weight $\omega(\sigma)\in \{-1,0,1\}$.
\end{corollary}
\begin{proof}
Every 1-VASS coverability and reachability language can be realized by 1-VASS in which 
the weight of each transition lies in $\{-1,0,1\}$ (general transitions can be handled by inserting intermediate states and $\varepsilon$-transitions).
Applying the reduction in the proof of Proposition~\ref{pro:VASStoMON} to such a 1-VASS leads to instances of 
 \textsc{Cover Closure} and \textsc{Reach Closure} in which
the weight map $\omega:\Sigma^* \rightarrow \mathbb Z$ satisfies 
$\omega(\sigma)\in\{-1,0,+1\}$ for all $\sigma \in \Sigma$.
\end{proof}

\section{Background on Exterior Algebra}
\label{app:ext}

Let $V=\mathbb Q^d$ be the vector space over the field $\mathbb Q$, its \emph{exterior algebra}  $\ExtAlg{V}$ is a graded $\mathbb Q$-vector space equipped with an associative, bilinear, and antisymmetric wedge product~$\wedge:\ExtAlg{V}\times\ExtAlg{V}\to  \ExtAlg{V}$. It decomposes as 
\[
\ExtAlg{V}=\bigoplus_{r=0}^n\ExtAlg[r]{V},
\]
 where $\ExtAlg[r]{V}$ is the $r^{\text{th}}$ exterior power, spanned by wedge product $v_1\wedge \cdots \wedge v_r$ with $v_i\in V$.
Let $e_1,\ldots,e_d$ denotes the standard basis of $\mathbb Q^d$, then a basis of $\ExtAlg[r]{V}$ is given by $e_{i_1}\wedge\cdots\wedge e_{i_r}$ with $1\leq i_1<\cdots<i_r\leq d$, and $\dim\ExtAlg[r]{V}=\binom{d}{r}$ (with $\binom{d}{r}=0$ for $r>d$). Moreover, a wedge product $v_1\wedge \cdots\wedge v_r$ is nonzero if and only if the vectors $v_1,\ldots,v_r$ are linearly independent.

The Grassmanian $\mathrm{Gr}(r,d)$ is the set of all $r$-dimensional subspaces of $\mathbb Q^d$.There is an injective map 
\[
\iota:\mathrm{Gr}(r,d)\to \mathbb P(\ExtAlg[r]{V})
\]
defined by mapping a subspace $W\in \mathrm{Gr}(r,d)$ with basis $v_1,\dots,v_r$ to $\iota(W)=[v_1\wedge \cdots \wedge v_r]$. The map is well-defined since if $v_1\wedge \cdots \wedge v_r=\alpha u_1\wedge\cdots\wedge u_r$ for some scalar $\alpha\in \mathbb Q\setminus \{0\}$, so $\iota(W)$ is defined up to scalar multiplication.

For subspaces $W_1,W_2\subseteq \mathbb Q^d$, we have $W_1\cap W_2=\{0\}$ if and only if $\iota(W_1)\wedge \iota(W_2)\neq 0$.

\section{Missing Proofs of~\Cref{sec:factorisation}}
\label{app-factorisation}

\clmrankproduct*
 \begin{proofclaim}
For an arbitrary \(i \leq k\), the product \(M_1 \cdots M_m\) can be divided as 
  $$M_1 \cdots M_m = (M_1 \cdots M_{i-1})(M_i \cdots M_k)(M_{k+1} \cdots M_m).$$
Since $\text{rank}(M_1 \cdots M_m) = r \leq \text{rank}(M_i \cdots M_k) \leq \text{rank}(M_i) = r$, all segments of the product must have rank \(r\).

To prove the second part, observe the following identity for any pair of matrices $A,B \in M_{d\times d}(\mathbb Q)$:
\[ \rank(AB) = \rank(B) - \dim(\ker(A) \cap \im(B)).\]
Applying this to $M_k M_{k+1}$, and using that each $M_i$ has rank $r$, we have:
 \[r = \rank(M_k M_{k+1}) = \rank(M_{k+1}) - \dim(\ker(M_k) \cap \im(M_{k+1})).
\]
 Hence, $\dim(\im(M_{k+1}) \cap \ker(M_{k})) = 0$, and therefore $\ker(M_k) \cap \im(M_{k+1}) = \{0\}$.  

\end{proofclaim}

\section{Missing Proofs of~\Cref{sec:background}}
\label{app:background}

\factclosureclosure*
\begin{proof}
Since $S \subseteq \overline{S}$, we have $f(S) \subseteq f(\overline{S})$; taking closures gives $\overline{f(S)} \subseteq \overline{f(\overline{S})}$.
On the other hand, by definition of continuity we have $f(\overline{S})\subseteq \overline{f(S)}$.  Taking closures gives the reverse inclusion.
\end{proof}

\factclosedsemigroupissubgroup*
\begin{proof}
    Let $S$ be a closed semigroup of $GL_d(\overline{\mathbb Q})$.
    It suffices to argue  that 
     $S$ is stable under matrix inversion.
     Given $M \in S$, clearly $MS \subseteq S$ holds as $S$ is closed under matrix multiplication; moreover, since $M$ is invertible~$MS$ is a closed set. 
     Observe that the sequence 
   \begin{equation}
   \label{eq-seq}
       S\supseteq MS\supseteq M^2S\supseteq\cdots 
   \end{equation}
   is a decreasing sequence of closed sets.
    By Noetherian property of ring of polynomial, \eqref{eq-seq} implies that there exists $k\in \mathbb N$ such that $M^k\,S=M^{k+1}S $. Since $M$ is invertible, the above shows that $S=MS$, implying that~$I_n\in MS$. This concludes the proof.
\end{proof}

\lemclosureisgroup*

\begin{proof}
Since $M$ is stable we have $\im(M) \cap \ker(M) = \{0\}$. Let $r = \rank(M)$, and 
let $V = \im(M)$ and $U = \ker(M)$, and choose a basis $B_V = \{e_1, \dots, e_r\}$ of $V$ and a basis $B_U = \{e_{r+1}, \dots, e_d\}$ of $U$. Let $Y \in GL_d(\mathbb Q)$ be the matrix whose columns are the vectors $e_1, \dots, e_d$. Then
\[
Y^{-1} M Y = \begin{pmatrix} A & 0 \\ 0 & 0 \end{pmatrix}
\]
for some $A \in GL_r(\mathbb Q)$. 
The powers of $M$ satisfy
\[
Y^{-1} M^k Y = \begin{pmatrix} A^k & 0 \\ 0 & 0 \end{pmatrix}.
\] 
Let $\langle  A \rangle \subset GL_r(\mathbb Q)$ be the  semigroup generated by $A$. By Fact~\ref{fact:closed_semigroup_is_subgroup}, $\overline{\langle A \rangle } \cap GL_r(\overline{\mathbb Q}) $ is a subgroup of  $GL_r(\overline{\mathbb Q})$. Since the conjugation map $X \mapsto Y X Y^{-1}$ is a homeomorphism, it commutes with taking closures. Hence,
\[
\overline{\langle M \rangle} = Y \left\{ \begin{pmatrix} B & 0 \\ 0 & 0 \end{pmatrix} \;\middle|\; B \in \overline{\langle A\rangle }\right\} Y^{-1}.
\] 
Intersecting with $\{X \in M_d(\overline{\mathbb Q}) \mid \rank(X) = r\}$ restricts $B$ to be invertible:
\[
G \coloneqq  \overline{\langle M \rangle} \cap \{X \in M_d(\overline{\mathbb Q}) \mid \rank(X) = r\}
= Y \left\{ \begin{pmatrix} B & 0 \\ 0 & 0 \end{pmatrix} \;\middle|\; B \in \overline{\langle A \rangle} \cap GL_r(\overline{\mathbb Q}) \right\} Y^{-1}.
\] 
 Thus, since $\overline{\langle A \rangle} \cap GL_r(\overline{\mathbb Q})$ is a group, we conclude that $G$ is also a group.

\end{proof}



\section{Missing Proofs of~\Cref{sec:indexG}}
\label{app:indexG}

\clmtarnsred*
\begin{proofclaim}
By~\Cref{fact:outerprodcut}, $ \bm y\otimes \bm 1$ is a rank-one matrix with repeated column~$\bm{y}=(yx_1,\dots, yx_n,y)^\top$.
 Observe that ${(M(\bm{y}\otimes \bm 1))}_{i:}=(M_{i:}\,\bm{y}) \bm 1^\top$. 
We now consider each case:
\begin{itemize}
    \item If $M_{i:}=\bm{0}^\top$, then $\bm 0^\top\bm y=0$. So ${(M(\bm{y}\otimes \bm 1))}_{i:}=\bm 0$.
    \item If $M_{i:}=(\bm {e}_i\pm \bm {e}_{n+1})^\top$, then $ M_{i:}\,\bm y=yx_i\pm y$. Hence, ${(M(\bm{y}\otimes \bm 1))}_{i:}=(yx_i\pm y)\bm 1^\top$.
    \item If $M_{i:}=\bm e_i^\top$, then $ M_{i:}\,\bm y=yx_i$, giving ${(M(\bm{y}\otimes \bm 1))}_{i:}=(yx_i)\bm 1^\top$.
\end{itemize}
\end{proofclaim}

\color{black}
\clindexgrammar*
   \begin{proofclaim}
Recall that given a sequence of transitions $\delta_1,\dots,\delta_k$ the index grammar produces words of the following form
\begin{equation*}
w= w_k \,  t_{\delta_k} \, w_k \, t_{\delta_{k-1}} \, w_{k-1} \, t_{\delta_{k-2}} \, w_{k-2}  \cdots w_3 \,  t_{\delta_2} \, w_2 \, t_{\delta_1} w_1
\end{equation*}
where $w_i=\delta_i w_{i-1}=\delta_i\delta_{i-1}\cdots \delta_0$ for $i \in \{k,\ldots,2\}$; see~\eqref{eq-derv2} and~\eqref{eq-w-unde}. 
For such sequence of transitions~$\delta_1,\dots,\delta_k$ and the associated word
$w$, we prove the following statement by induction on $k$:
\begin{itemize}
    \item if the sequence $(q_0,\bm 0)\xrightarrow{\delta_1} \dots \xrightarrow{\delta_k}(q_k,\bm x_k) $ does not correspond to a run in $\mathcal V$, then $\varphi(w)=\bm 0$.
    This can happen due   to the state mismatch  (that is, for some transition $\delta_i$, the source state  is not equal to the target state of the preceding transition $\delta_{i-1}$)  or a counter value~$\bm{x}_i$ not being nonnegative. 
    \item if the sequence corresponds to a run reaching configuration $(q_k,\bm x_k)$, then the matrix $\varphi(w)$ has exactly one non-zero block with identical columns $y(\bm x_k,1)^\top$, for some $y\in\mathbb Q$. 
\end{itemize}

Base Case ($k=0$): The run consists of only the initial configuration $(q_0,\bm 0)$. The grammar produces the word $w=\delta_0$. By definition, $\varphi(\delta_0)$ is a matrix with a single non-zero block $B= (\bm 0,1)\otimes \bm 1$ at index $(q_0,q_0)$. This matches the hypothesis for a run where $q_k=q_0, \bm x_k=\bm 0$. 

 Inductive step ($k>0$):
Assume the hypothesis holds for any sequence of length $k-1$. We now analyze the word $w$ for the sequence $\delta_1,\dots, \delta_k$, where $\delta_k=(p_k,\bm v,q_k)$. 
We consider the following distinguished cases for the~$k$-th step. Write $w'$ to be the word assigned to $\delta_1, \dots, \delta_{k-1}$, \emph{i.e.}
 \[w'= w_{k-1} \, t_{\delta_{k-1}} \, w_{k-1} \, t_{\delta_{k-2}} \, w_{k-2}  \cdots w_3 \,  t_{\delta_2} \, w_2 \, t_{\delta_1} w_1.\]

\begin{itemize}
    \item Case 1: A proper prefix of the run is invalid.
    If the run $\delta_1,\dots,\delta_{k-1}$ is invalid, then by inductive hypothesis~$\varphi(w')=\bm 0$. 
     Since $w=w_kt_{\delta_k}\delta_k w'$, then $\varphi(w)=\varphi(w_k t_{\delta_k}\delta_k)\varphi(w')=\bm 0$, as required.
    \item Case 2: The run is invalid at step $k$ due to state mismatch.
    Suppose $\delta_1,\dots,\delta_{k-1}$ is a run to $(q_{k-1},\bm x_{k-1})$, but the transition $\delta_k=(p_k,\bm v_k,q_k)$ cannot be applied because its source state $p_k$ does not match the current state $q_{k-1}$.

    By the inductive hypothesis, $\varphi(w')$ has its single non-zero block at index $(q_{k-1},q_0)$. The construction of~$w$ by the grammar ensures a matrix multiplication of the form $X\varphi(\delta_k)\varphi(w_{k-1})Y$. In this product, the column-state index $q_{k-1}$ of $\varphi(w_{k-1})$ must align with the row-state index $p_k$ of $\varphi(\delta_k)$. Since $p_k\neq q_{k-1}$, this block matrix multiplication results in the $\bm 0$ matrix. Thus, $\varphi(w)=\bm 0$ as required.

   \item Case 3: The run is invalid at step $k$ due to $\bm{x}_k$ not being nonneagtive. 
   Suppose the run is valid up to $\delta_1\dots \delta_{k-1}$ and leads to $(q_{k-1},\bm x)$, but applying $\delta_k$ it becomes invalid, because $\bm x+\bm v_k$ has a negative component in the $\ell$-th coordinate. By induction hypothesis, $\varphi(w')$ has exactly one non-zero block  $y((\bm{x},1))\otimes\bm{1}$, positioned at index $(q_{k-1},q_0)$.
   One can observe that there exists $c_1\in \mathbb Q$, such that $\varphi (\delta_1\delta_0 t_{\delta_1}\delta_1\delta_0)=c_1\varphi(\delta_1\delta_0)$. Therefore 
\begin{align*}
\varphi(w_2t_{\delta_2}w_2t_{\delta_1}w_1)&=\varphi(w_2t_{\delta_2}\delta_2\delta_1\delta_0 t_{\delta_1}\delta_1\delta_0)\\
&=\varphi (w_2t_{\delta_2})\varphi(\delta_2)c_1\varphi(\delta_1\delta_0)\\
&=c_1\varphi (w_2)\varphi(t_{\delta_2})\varphi(w_2).
   \end{align*}
 Inductively, there exists $c\in \mathbb Q$ such that
 \[\varphi(w)=c\varphi(w_k)\varphi(t_{\delta_k})\varphi(w_k).
 \]
 Recall that for the letter~$t_{\delta_k}$,  we define  $\varphi(t_{\delta_k})$ to be a matrix with the nonzero block $B$ at index $(q,q_0)$, where $B:=\bm 1\otimes (\bm e_\ell+\bm e_{n+1})$. By~\Cref{clm-updatedelta},
the product $\varphi (w)=c\varphi(w_k)\varphi(t_{\delta_k})\varphi(w_k)$
results in a matrix with exactly one non-zero block whose columns are all~$c(n+1)y^2(x_\ell+1)(\bm{x},1)$. 
If $x_\ell=-1$, this product is zero. 
\item Case 4: The run $\delta_1, ..., \delta_k$ is valid.
A similar proof to Case 3 establishes the induction statement. 
\end{itemize}
The above induction gives that if $\mathrm{Reach(\mathcal{V})}$ is finite then 
$\overline{\varphi(L(G))}$  is contained in a finite union of lines and thereby has dimension at most
one. 

For the reverse direction, assume that  $\mathrm{Reach(\mathcal{V})}$ is infinite, we show that  $\overline{\varphi(L(G))}$ has dimension strictly more than~$1$. 
The proof follows by a similar argument presented in~\cite[Proof of Proposition 17]{HrushovskiOPW23}. The proof relies on~\cite[ Corollary 4]{cox1997ideals} stating that 
given an  affine variety $V \subseteq \overline{Q^n}$,  the dimension of $V$ is equal to
the largest integer~$r$ for which there exist r variables $x_{i_1} , \ldots, x_{i_r}$ such that the vanishing ideal $I(V) \cap
\mathbb{Q}[x_{i_1} , \ldots, x_{i_r}] = {0}$, meaning that~$I(V)$ contain no nonzero polynomials in only
these variables. 

The argument  follows by the fact that if $\mathrm{Reach(\mathcal{V})}$ is infinite, there exists some transition $\delta$ from state $p$ to~$q$, for some $p,q\in Q$, and some coordinate $i\in \{1,\ldots,n\}$ such that  there are infinitely many runs in $\mathcal{V}$ taking $\delta$ as the last transition. 
Hence, there are infinitely many matrices  $M\in \varphi(L(G))$ with the nonzero block at index~$(q,q_0)$ with growing entries in the $i$-th row.
Recall that all entries in each row are equal and are always a multiple of entries   in the $n+1$-th row ((or simply unit). Assume that  entries in the $i$-th row is~$\ell$, then the unit is  larger than $2^\ell$, as after each transition in~$\mathcal{V}$ a matrix $\varphi(t_{\delta'})$ was produced to check if the prefix of the run is valid, and this matrix at least doubles the value of 
the unit. This implies that there is no polynomial relations between the entries of $i$-th and $(n+1)$-th rows in the $(q,q_0)$-th block   without mentioning other variables  in the closure. 

\end{proofclaim}

\end{document}